\documentclass[11pt]{article}
\usepackage{amssymb}
\usepackage{amsthm}
\usepackage{amsmath}
\usepackage[dvips]{epsfig}
\usepackage{graphicx}
\usepackage{color}
\usepackage{url}

\newtheorem{theorem}{Theorem}
\newtheorem{lemma}{Lemma}
\newtheorem{corollary}{Corollary}

\newtheorem{remark}{Remark}

\newcommand{\R}{\mathbb{R}}
\newcommand{\C}{\mathbb{C}}
\newcommand{\Z}{\mathbb{Z}}
\newcommand{\N}{\mathbb{N}}

\begin{document}

\title{\bf Nonlinear stationary states in PT-symmetric lattices}

\author{Panayotis G. Kevrekidis$^{1}$, Dmitry E. Pelinovsky$^{2,3}$, and Dmitry Y.Tyugin$^{3}$ \\
$^{1}$ {\small \it Department of of Mathematics and Statistics, University of
Massachusetts, Amherst, MA 01003-9305, USA} \\
$^{2}$ {\small \it Department of Mathematics and Statistics, McMaster
University, Hamilton, Ontario, Canada, L8S 4K1 } \\
$^{3}$ {\small \it Department of Applied Mathematics, Nizhny Novgorod State
Technical University, Nizhny Novgorod, Russia }}

\date{\today}
\maketitle

\begin{abstract}
In the present work we examine both the linear and nonlinear properties
of two related PT-symmetric systems of the discrete nonlinear Schr{\"o}dinger (dNLS)
type.

First, we examine the parameter range for which the finite PT-dNLS chains
have real eigenvalues and PT-symmetric linear eigenstates.
We develop a systematic way of analyzing the nonlinear stationary states
with the implicit function theorem at an analogue of
the anti-continuum limit for the dNLS equation.

Secondly, we consider the case when a finite PT-dNLS chain is
embedded as a defect in the infinite dNLS lattice. We show that the stability
intervals of the infinite PT-dNLS lattice are wider than in the case of
a finite PT-dNLS chain. We also prove existence of localized
stationary states (discrete solitons) in the analogue of the anti-continuum limit for the dNLS equation.

Numerical computations illustrate the existence of nonlinear stationary states,
as well as the stability and saddle-center bifurcations of discrete solitons.
\end{abstract}


\section{Introduction}

The subject of PT-symmetry and its physical implications has
gained a tremendous momentum over the past few years. This
field was initiated by the original proposal of C. Bender~\cite{R1}
who suggested that the linear Schr\"{o}dinger operator with a
complex-valued potential, which is symmetric with respect to combined
parity (P) and time-reversal (T) transformations, is guaranteed
to have real spectrum at a certain parametric regime. Thus, this was
proposed as a viable alternative for the standard Hermitian quantum
mechanics. Yet, it was the pioneering work in the group of D.
Christodoulides both at the theoretical~\cite{ziad,Ramezani} and
experimental~\cite{dncnat} levels that
showcased nonlinear optics as a fertile ground for the physical
implementation of the PT-symmetric potentials. These efforts have
motivated a wealth of recent works, especially on the physical
side, addressing various aspects of continuous and discrete PT-symmetric
systems. These include among others the study of
the fragility of PT-symmetry in linear problems~\cite{Bendix,Pelin2},
nonlinear stationary states of few site configurations (also referred
to as oligomers, or plaquettes in two-dimensional lattices)~\cite{Li,Guenter,suchkov,Sukhorukov,ZK},
as well as solitary waves and breathers in infinite systems both
continuous~\cite{Abdullaev,baras2,baras1,malom1,malom2,Nixon} and discrete \cite{Dmitriev,Pelin1,Sukh}.

While the number of studies of such PT-symmetric systems both in
optics \cite{Kirillov} and in atomic physics~\cite{graefe,tsampikos} is rapidly growing,
the volume of related mathematical works is rather limited
and mostly constrained to linear problems~\cite{Azizov,trunk,mostafaz,Weigert}.
It is the purpose of this paper to provide a number of rigorous results on
nonlinear stationary states in PT-symmetric discrete systems. Our emphasis
will be two-fold.

First, we will consider finite PT-symmetric chains of the discrete nonlinear
Schr{\"o}dinger (dNLS) type~\cite{dnls}. We examine their phase transitions (from a PT-symmetric oscillatory
phase to the exponentially growing phase) when the gain and loss parameter is increased.
The nonlinear stationary states bifurcate from the linear PT-symmetric states
by means of a standard local bifurcation. On the other hand, we will also consider large-amplitude
stationary states in an analogue of the well-known anti-continuum limit for the dNLS equation~\cite{pkf,PelSak}, through a suitable
rescaling of the PT-dNLS equation. This rescaling enables us to use the implicit function theorem to continue
stationary states from the limit, where they are effectively uncoupled and the gain and loss
parameter is negligibly small.

Second, we consider the case where the finite PT-symmetric
chains are embedded in the infinite nonlinear lattice of the dNLS type. Again,
we will examine phase transitions of such systems and will  prove
that the infinite PT-dNLS lattice has a wider stability interval
compared to the isolated PT-dNLS chains.
We also develop a proof of the existence of localized stationary states
(discrete solitons) in the PT-dNLS equation.
Numerical computations illustrate the theoretical results on
existence of nonlinear stationary states, as well as the stability and saddle-node
bifurcations of discrete solitons.

Note that our technique allows us to prove existence of discrete solitons in the infinite PT-dNLS
equation, but such discrete solitons are unstable
because the phase transition in this infinite lattice occurs already at the zero value of
the gain and loss parameter \cite{Pelin2}. Earlier, existence of such discrete
solitons was observed in numerical continuations from the diatomic PT-dNLS lattice \cite{Pelin1}.

The article is structured as follows. Section 2 covers fundamentals of the PT-dNLS
equaiton. Section 3 is devoted to finite dNLS chains with four sections
on eigenvalues of the linear PT-dNLS equation, local bifurcations of stationary states, bifurcations of
large-amplitude stationary states, and numerical results. Section 4 is
concerned with the PT-symmetric defects in infinite dNLS lattices and contains three sections
on eigenvalues of the linear PT-dNLS equation, bifurcations of discrete solitons from the anti-continuum
limit, and numerical results. Section 5 concludes the article with a summary and a
discussion of future directions

{\bf Acknowledgments:} The authors thank James Dowdall for help at an early stage of
the project during his NSERC USRA work and Dimitri Frantzeskakis for 
discussions on the subject of PT-symmetry. The work of P.K. is partially supported by the US
National Science Foundation under grants NSF-DMS-0806762, and NSF-CMMI-1000337, from the Alexander
von Humboldt Foundation and from the US AFOSR under grant FA9550-12-1-0332.
The work of D.P. is supported in part by NSERC and by the ministry of education
and science of Russian Federation (Project 14.B37.21.0868).

\section{Formalism of the PT-dNLS equation}

We consider the discrete nonlinear Schr\"{o}dinger (dNLS) equation with non-conservative
terms that introduce gains and losses of nonlinear oscillators. When gains and losses are
combined in a compensated network, the model referred to as the PT-dNLS equation
takes the form
\begin{equation}
\label{dnls}
i \frac{d u_n}{d t} = u_{n+1} - 2 u_n + u_{n-1} + i \gamma (-1)^n u_n + |u_n|^2 u_n,
\end{equation}
where parameter $\gamma$ stands for the gain and loss coefficient. The finite 
PT-dNLS chain is defined for $n \in \{1,2,...,2N\}$ for a positive integer $N$
subject to Dirichlet boundary conditions $u_0 = u_{2N+1} = 0$,
whereas the infinite PT-dNLS lattice is defined for all integers $n$ on $\mathbb{Z}$
subject to the decay of $u_n$ to zero as $|n| \to \infty$.
The amplitudes $u_n$ for all admissible values of $n$ are
complex-valued functions of time $t$.

For notational consistency, we denote the sequence $\{ u_n \}_{n \in \mathbb{Z}}$
of complex-valued amplitudes $u_n$ by the vector notation ${\bf u}$.
These vectors are considered in Hilbert
space $l^2(\Z)$ equipped with the inner product $\langle {\bf u}, {\bf v} \rangle := \sum_{n \in \mathbb{Z}}
\bar{u}_n v_n$ and the induced norm $\| {\bf u} \| := \left( \sum_{n \in \mathbb{Z}} |u_n|^2 \right)^{1/2}$.

Let us formulate the evolution problem (\ref{dnls}) in the complex Hamiltonian form
\begin{equation}
\label{dyn-sys}
i \frac{d {\bf u}}{d t} = \nabla_{\bar{\bf u}} H({\bf u},\bar{\bf u}),
\end{equation}
where ${\bf u}$ is a collection of amplitudes $u_n$ for all admissible values of $n$ (denoted by $S$)
in some function space (denoted by $X$), the bar denotes complex conjugation,
and the complex-valued Hamiltonian functional $H : X \to \C$ takes the form
\begin{equation}
\label{complex-Hamiltonian}
H({\bf u},\bar{\bf u}) = - \sum_{n \in S} |u_{n+1}-u_n|^2 +
i \gamma \sum_{n \in S} (-1)^n |u_n|^2 + \frac{1}{2} \sum_{n \in S} |u_n|^4.
\end{equation}

The dynamical system (\ref{dyn-sys})
is said to be $PT$-symmetric if there is a linear real-valued
$t$-independent operator $P : X \to X$ such that
\begin{equation}
\label{nonlinear-P}
P^2 = I \quad \mbox{and} \quad \bar{\bf N}({\bf u},\bar{\bf u}) = P {\bf N}(P {\bf u},P \bar{\bf u}),
\end{equation}
where ${\bf N}({\bf u},\bar{\bf u}) := \nabla_{\bar{\bf u}} H({\bf u},\bar{\bf u})$
and $I : X \to X$ is an identity operator.

If ${\bf u}(t)$ is a solution of the PT-symmetric dynamical system (\ref{dyn-sys}) for $t$
in a symmetric interval $J := (-t_0,t_0) \subset \R$ with some positive $t_0$,
then ${\bf v}(t) := P \bar{\bf u}(-t)$ is another solution of the same system for $t \in J$.
This statement can be checked by direct substitution. This symmetry suggests
the following definition of the operator $T : C(J,X) \to C(J,X)$:
\begin{equation}
T {\bf u}(t) := \bar{\bf u}(-t), \quad t \in J.
\end{equation}
Note that the operator $T$ is sesquilinear in ${\bf u}$ and nonlocal in $t$.
The letters $P$ and $T$ stand for {\em parity} and {\em time reversal} transformations,
which correspond to fundamental symmetries in physics.

When the vector field ${\bf N}$ is linear and given by
${\bf N}({\bf u},\bar{\bf u}) := \mathcal{H} {\bf u}$ associated
with a linear complex-valued bounded operator $\mathcal{H} : X \to X$,
then the $PT$-symmetry is expressed in the standard form
\begin{equation}
\label{linear-P}
\bar{\mathcal{H}} = P \mathcal{H} P.
\end{equation}

Our first result is to show that the dNLS equation with compensated gain and loss terms
(\ref{dnls}) is a PT-symmetric dynamical system both for finite and infinite chains.

\begin{lemma}
Define $S := \{ 1,2,...,2N \}$ for a positive integer $N$ and $X := \C^{2N}$.
Then, the dynamical system (\ref{dyn-sys}) with complex-valued Hamiltonian (\ref{complex-Hamiltonian})
is PT-symmetric with respect to the operator $P : \C^{2N} \to \C^{2N}$ given by
$$
[P {\bf u}]_n = u_{2N+1-n}, \quad n \in S.
$$
\label{lemma-PT-symmetry}
\end{lemma}

\begin{proof}
We verify the statement with the explicit computation. Given the definition of $P$,
we obtain
$$
[P^2 {\bf u}]_n = [P {\bf u}]_{2N+1-n} = u_n
$$
and
$$
[{\bf N}(P {\bf u},P \bar{\bf u})]_n = u_{2N + 2 - n} - 2 u_{2N+1-n} + u_{2N-n}
+ i \gamma (-1)^n u_{2N+1-n} + |u_{2N+1-n}|^2 u_{2N+1-n}.
$$
Applying $P$ again, we obtain
\begin{eqnarray*}
[P {\bf N}(P {\bf u},P \bar{\bf u})]_n & = & u_{n+1} - 2 u_{n} + u_{n-1}
+ i \gamma (-1)^{2N+1-n} u_{n} + |u_{n}|^2 u_{n} \\
& = & u_{n+1} - 2 u_{n} + u_{n-1}
- i \gamma (-1)^{n} u_{n} + |u_{n}|^2 u_{n} \\
& = & [\bar{\bf N}({\bf u},\bar{\bf u}) ]_n,
\end{eqnarray*}
which recovers the second identity (\ref{nonlinear-P}).
\end{proof}

\begin{remark}
The symmetry of Lemma \ref{lemma-PT-symmetry} can be proven
by simple reflection arguments. If the chain of oscillators has
the damped site at the left end and the gained site at the right end, then
since $P$ reflects all oscillators about the middle point, the reflected chain has now
the gained site at the left end and the damped site at the right end, that is,
the reflected chain is equivalent to the complex conjugate chain.
\end{remark}

\begin{remark}
Although $P$ in Lemma \ref{lemma-PT-symmetry} represents the fundamental physical symmetry,
other choices of operator $P$ are possible for the linear terms of the dynamical 
system (\ref{dyn-sys})--(\ref{complex-Hamiltonian}).
For instance, if $N = 2$, there exists another operator $P_a$
such that $P_a^2 = {\rm id}$ and $\bar{\mathcal{H}} = P_a \mathcal{H} P_a$, where
$$
\mathcal{H} = \left[ \begin{array}{cccc} -2 - i \gamma & 1 & 0 & 0 \\
1 & -2 + i \gamma & 1 & 0 \\
0 & 1 & -2 - i \gamma & 1  \\
0 & 0 & 1 & -2 + i \gamma \end{array} \right], \quad
P_a = \left[ \begin{array}{cccc} 0 & -2a & 0 & a \\
-2a & 0 & -a & 0 \\
0 & -a & 0 & -2a  \\
a & 0 & -2a & 0 \end{array} \right]
$$
with either $a = \frac{1}{\sqrt{5}}$ or $a = -\frac{1}{\sqrt{5}}$
(this statement can be easily checked by means of symbolic software).
Nevertheless, the operator $P_a$ does not represent the PT-symmetry 
of the full nonlinear system (\ref{dyn-sys})--(\ref{complex-Hamiltonian})
because the nonlinear term ${\bf N}_{\rm non}({\bf u},\bar{\bf u}) :=
{\bf N}({\bf u},\bar{\bf u}) - \mathcal{H} {\bf u}$
does not satisfy the second identity (\ref{nonlinear-P}).
For instance, we have
\begin{eqnarray*}
[P {\bf N}_{\rm non}(P {\bf u},P \bar{\bf u})]_1 & = & a^4 \left[ 2 |2u_1+u_3|^2(2u_1+u_3) +
|u_1-2u_3|^2 (u_1 - 2u_3) \right] \\
& = & \frac{1}{25} \left[ 17 |u_1|^2 u_1 + 12 |u_1|^2 u_3 + 8 u_3^2 \bar{u}_1 + 6 u_1^2 \bar{u}_3 + 16 |u_3|^2 u_1 - 6 |u_3|^2 u_3 \right]\\
& \neq & |u_1|^2 u_1 = [\bar{\bf N}_{\rm non}({\bf u},\bar{\bf u})]_1,
\end{eqnarray*}
hence the second identity (\ref{nonlinear-P}) is not satisfied.
\end{remark}

\begin{corollary}
Let $S = \Z$ and $X = l^2(\Z,\C)$. For any fixed $n_0 \in \Z$, the dynamical system (\ref{dyn-sys})
with complex-valued Hamiltonian (\ref{complex-Hamiltonian})
is PT-symmetric with respect to the operator $P : l^2(\Z) \to l^2(\Z)$ given by
$$
[P {\bf u}]_n = u_{n_0-n}, \quad n \in \Z.
$$
\label{corollary-PT-symmetry}
\end{corollary}

\begin{proof}
The proof follows from the proof of Lemma \ref{lemma-PT-symmetry} when
$S = \{ 1,2,...,2N \}$ is replaced by $S = \Z$ and the value of $n_0 \in \Z$ is arbitrary.
\end{proof}

\section{Finite PT-dNLS lattices}

We shall now consider the PT-dNLS equation (\ref{dnls}) for the finite chain $S_N := \{1,2,...,2N\}$,
 where $N$ is a positive integer, subject to the Dirichlet boundary conditions $u_0 = u_{2N+1} = 0$.
We study $2N$ eigenvalues of the linear PT-dNLS equation to find the phase transition threshold
$\gamma_N$, which separates the neutral stability of the zero solution for 
$\gamma \in (-\gamma_N,\gamma_N)$ and the linear instability of the zero solution 
for $|\gamma| > \gamma_N$. We show that $\gamma_N$ is a monotonically
decreasing sequence of $N$ such that $\gamma_1 = 1$ and $\gamma_N \to 0$ as $N \to \infty$.

We consider local bifurcations of nonlinear stationary states of the PT-dNLS equation (\ref{dnls})
from the linear limit and prove that every
simple eigenvalue of the linearized PT-dNLS equation generates a unique (up to a gauge transformation)
family of the PT-symmetric stationary states in the parameter space. For $\gamma$ inside
the stability interval $(-\gamma_N,\gamma_N)$,
this yields the existence of $2N$ branches of stationary states. These $2N$ branches are
extended towards a large-amplitude limit with some intermediate bifurcations.

We characterize the number and properties of the branches of the stationary states
in the large-amplitude limit and show that there exist $2^{N}$ distinct branches for
any $\gamma \in (-\gamma_1,\gamma_1) = (-1,1)$, for which $|u_n|^2$
is large for all $n \in S_N$. We also discuss
existence of other branches of the stationary
states, which are centered at the middle sites of $S_N$ and
for which $|u_1|^2$ is small in the large-amplitude limit.

These analytical results are illustrated with numerical approximations of the nonlinear stationary states
of the PT-dNLS equation (\ref{dnls}) for the finite chain with $N = 1,2,3$.

\subsection{Eigenvalues of the linear PT-dNLS equation}

We consider the linear stationary PT-dNLS equation on a finite chain $S_N := \{1,2,...,2N\}$:
\begin{equation}
\label{dnls-linear}
E w_n = w_{n+1} + w_{n-1} + i \gamma (-1)^n w_n, \quad n \in S_N,
\end{equation}
subject to the Dirichlet boundary conditions $w_0 = w_{2N+1} = 0$.
Compared to the PT-dNLS equation (\ref{dnls}), the diagonal term of the discrete Laplacian operator
has been included in the definition of the parameter $E$ (see Remark \ref{remark-eigenvalue} below).
We shall find all $2N$ eigenvalues of the linear stationary dNLS equation (\ref{dnls-linear}) in explicit form, a result
from which the phase transition threshold $\gamma_N$ is computed also explicitly.

\begin{theorem}
Eigenvalues of the linear eigenvalue problem (\ref{dnls-linear}) are found explicitly
from the set of quadratic equations:
\begin{equation}
\gamma^2 + E^2 = 4 \cos^2\left(\frac{\pi j}{1 + 2N}\right), \quad 1 \leq j \leq N.
\end{equation}
In particular, all eigenvalues are simple and real for $\gamma \in (-\gamma_N,\gamma_N)$,
where
\begin{equation}
\gamma_N := 2 \cos\left(\frac{\pi N}{1 + 2N} \right).
\end{equation}
\label{theorem-linear-PT}
\end{theorem}

\begin{proof}
By writing
$$
x_k = w_{2k-1}, \quad y_k = w_{2k}, \quad 1 \leq k \leq N,
$$
we can rewrite the linear eigenvalue problem (\ref{dnls-linear}) in the equivalent form:
\begin{equation}
\label{dnls-linear-equiv}
\left\{ \begin{array}{l}
E x_k = y_{k-1} + y_{k} - i \gamma x_k, \\
E y_k = x_{k} + x_{k+1} + i \gamma y_k, \end{array} \right. \quad 1 \leq k \leq N,
\end{equation}
where the boundary conditions are now $y_0 = 0$ and $x_{N+1} = 0$. Expressing
$y_k$ from the second equation of the system (\ref{dnls-linear-equiv})
and substituting it to the first equation of the system, we obtain a second-order
difference equation
$$
(\gamma^2 + E^2) x_k = x_{k-1} + 2 x_k + x_{k+1}, \quad 1 \leq k \leq N,
$$
where the boundary conditions are now $x_0 = -x_1$ and
$x_{N+1} = 0$. Using the discrete Fourier transform, we represent the eigenvector
satisfying the boundary condition $x_{N+1} = 0$ in the form
$$
x_k = \sin \theta(N+1-k), \quad 1 \leq k \leq N.
$$
Parameter $\theta$ in the fundamental
interval $[0,\pi]$ defines uniquely the spectral parameter $z := \gamma^2 + E^2$
from the dispersion relation
\begin{equation}
\label{gamma-E-equation}
z := \gamma^2 + E^2 = 2 + 2 \cos \theta = 4 \cos^2 \frac{\theta}{2}.
\end{equation}
From the remaining boundary condition $x_0 + x_1 = 0$, we obtain
$$
\sin \frac{\theta (1 + 2N)}{2} \cos\frac{\theta}{2} = 0,
$$
where $\cos\frac{\theta}{2} \neq 0$ (since $\{ x_k \}_{k=1}^N$ must not be
identically zero). From the roots of $\sin \frac{\theta (1 + 2N)}{2}$, we
obtain the admissible values of $\theta$ as follows:
$$
\theta = \frac{2 \pi j}{1 + 2N}, \quad 1 \leq j \leq N,
$$
which yields the result by (\ref{gamma-E-equation}).
\end{proof}

\begin{remark}
For each eigenvalue $E$ of the linear stationary dNLS equation (\ref{dnls-linear})
with the eigenvector ${\bf w}$, there exists another eigenvalue $\bar{E}$ with
the eigenvector $P \bar{\bf w}$. This is an elementary consequence of the PT-symmetry,
which produces a new solution ${\bf v}(t) = P \bar{\bf u}(-t) = (P \bar{\bf w}) e^{-i(E-2)t}$ of the time-dependent
dNLS equation (\ref{dnls}) from the solution ${\bf u}(t) = {\bf w} e^{-i (E-2) t}$ of the same
equation. In particular,
if $E$ is a simple real eigenvalue (as in Theorem \ref{theorem-linear-PT}), then
the eigenvector ${\bf w}$ can be chosen to satisfy the PT-symmetry
\begin{equation}
{\bf w} = P \bar{\bf w} \quad \Rightarrow \quad w_{n} = \bar{w}_{2N+1-n}, \quad n \in S_N.
\end{equation}
\label{remark-eigenvalue}
\end{remark}

We list some numerical values of the phase transition thresholds:
\begin{eqnarray*}
\gamma_1 & = & 2 \cos \frac{\pi}{3} = 1, \\
\gamma_2 & = & 2 \cos \frac{2 \pi}{5} \approx 0.618, \\
\gamma_3 & = & 2 \cos \frac{3 \pi}{7} \approx 0.445.
\end{eqnarray*}
Note that $\lim_{N \to \infty} \gamma_N = 0$.

\subsection{Stationary states: local bifurcations}

We shall now consider nonlinear stationary states on a finite chain $S_N$, which satisfy the
nonlinear stationary PT-dNLS equation:
\begin{equation}
\label{dnls-nonlinear}
E w_n = w_{n+1} + w_{n-1} + i \gamma (-1)^n w_n  + |w_n|^2 w_n, \quad n \in S_N,
\end{equation}
subject to the Dirichlet boundary conditions $w_0 = w_{2N+1} = 0$.
We shall work in the space $X = \mathbb{C}^{2N}$.

Assuming that the linear stationary PT-dNLS equation (\ref{dnls-linear})
admits a simple real eigenvalue $E_0$
with the eigenvector ${\bf w}_0 \in X$, we shall prove the existence of a branch
of the PT-symmetric stationary states ${\bf w} \in X$ satisfying the nonlinear
stationary PT-dNLS equation (\ref{dnls-nonlinear}) for
$E$ in a one-sided neighborhood of $E_0$. The solution branch is unique
up to a gauge transformation: ${\bf w} \to e^{i \alpha} {\bf w}$, where
$\alpha \in \R$. This result corresponds to the
standard {\em local bifurcation} of the nonlinear state ${\bf w}$ from the linear eigenstate
${\bf w}_0$, which is complicated here due to the presence of the PT-symmetry.

The local bifurcation results were considered with
formal perturbation expansions by Zezyulin \& Konotop \cite{ZK}.
Here we give a rigorous version of the same result.

\begin{theorem}
Assume that $E_0$ is a simple real eigenvalue of the linear stationary PT-dNLS equation (\ref{dnls-linear})
with the PT-symmetric eigenvector ${\bf w}_0 = P \bar{\bf w}_0$ in $X = \mathbb{C}^{2N}$. Then, there exists a unique
(up to a gauge transformation) PT-symmetric solution ${\bf w} = P \bar{\bf w}$ of the nonlinear
stationary PT-dNLS equation (\ref{dnls-nonlinear}) for real $E > E_0$. Moreover,
the solution branch is parametrized by a small parameter $a$ such that
the map $\R \ni a \to (E,{\bf w}) \in \R \times X$ is $C^{\infty}$ and 
for sufficiently small $a$, there is a positive constant $C$ such that
\begin{equation}
\label{bound-local}
\| {\bf w} \|^2 + |E-E_0| \leq C a^2.
\end{equation}
\label{theorem-local-bifurcation}
\end{theorem}

\begin{proof}
We write the nonlinear stationary PT-dNLS equation (\ref{dnls-nonlinear}) in the abstract form
\begin{equation}
\label{nonlinear-1}
(E - \mathcal{H}) {\bf w} = {\bf N}_{\rm non}({\bf w}),
\end{equation}
where $\mathcal{H} : X \to X$ is the linear (matrix) operator associated with the
right-hand side of the linearized stationary PT-dNLS
equation (\ref{dnls-linear}) and ${\bf N}_{\rm non}({\bf w}) : X \to X$ is the cubic nonlinear part. We note that
according to our assumptions, we have
$$
{\rm Ker}(E_0 - \mathcal{H}) = {\rm span}({\bf w}_0), \quad
{\rm Ker}(E_0 - \mathcal{H})^+ = {\rm span}(P {\bf w}_0),
$$
where
$$
(E_0 - \mathcal{H})^+ = E_0 - \bar{\mathcal{H}} = P (E_0 - \mathcal{H}) P.
$$
Using the standard Lyapunov--Schmidt method, we write
\begin{equation}
\label{decomposition-local}
E = E_0 + \Delta, \quad {\bf w} = a {\bf w}_0 + {\bf u}, \quad \langle P {\bf w}_0, {\bf u} \rangle = 0,
\end{equation}
where $(\Delta, a, {\bf u}) \in \C \times \C \times X$ are determined from the nonlinear equations (\ref{nonlinear-1})
projected to ${\rm Ker}(E_0 - \mathcal{H})^+$ and ${\rm Ran}(E_0 - \mathcal{H})^+$.
Recall that by the Fredholm theory, ${\rm Ker}(E_0 - \mathcal{H})^+$ is
orthogonal to ${\rm Ran}(E_0 - \mathcal{H})$ so that ${\bf u} \in {\rm Ran}(E_0 - \mathcal{H})$.

The projection to ${\rm Ker}(E_0 - \mathcal{H})^+$ is written in the scalar form:
\begin{equation}
\label{nonlinear-2}
\Delta a \langle P {\bf w}_0, {\bf w}_0 \rangle = \langle P {\bf w}_0, {\bf N}_{\rm non}(a {\bf w}_0 + {\bf u}) \rangle.
\end{equation}
By the implicit function theorem, the projection to ${\rm Ran}(E_0 - \mathcal{H})^+$ (not written here) guarantees
the existence and uniqueness of a smooth ($C^{\infty}$) map from $(\Delta,a) \in \C^2$ to
${\bf u} \in {\rm Ran}(E_0 - \mathcal{H}) \subset X$. Moreover, for small values of $\Delta$ and $a$,
there is a positive constant $C$ such that
\begin{equation}
\label{bound-on-u}
\| {\bf u} \| \leq C (1 + |\Delta|) |a|^3.
\end{equation}

For $a = 0$, we have a unique zero solution ${\bf u} = {\bf 0}$ and the equation
(\ref{nonlinear-2}) is satisfied identically. In what follows, we assume $a \neq 0$.

We claim that $\langle P {\bf w}_0, {\bf w}_0 \rangle \neq 0$ under the assumption that
$E_0$ is a simple eigenvalue of $\mathcal{H}$. Indeed, if
$\langle P {\bf w}_0, {\bf w}_0 \rangle = 0$, there exists a generalized eigenvector
${\bf w}_1 \in X$ for the same eigenvalue $E_0$ from a solution of the inhomogeneous equation
$$
(E_0 - \mathcal{H}) {\bf w}_1 = -{\bf w}_0,
$$
which is a contradiction to the assumption that $E_0$ is a simple eigenvalue of $\mathcal{H}$.

Therefore, $\langle P {\bf w}_0, {\bf w}_0 \rangle \neq 0$. Then, there exists
a unique smooth map from $a \in \C$ to $\Delta \in \C$ solving the bifurcation equation
(\ref{nonlinear-2}). Moreover, for small values of $a$, there is a positive constant $C$ such that
\begin{equation}
\label{decomposition-final}
|\Delta \langle P {\bf w}_0, {\bf w}_0 \rangle - |a|^2 \langle P {\bf w}_0, {\bf N}_{\rm non}({\bf w}_0) \rangle | \leq C |a|^4.
\end{equation}
Note that both $\langle P {\bf w}_0, {\bf w}_0 \rangle$ and $\langle P {\bf w}_0, {\bf N}_{\rm non}({\bf w}_0) \rangle$
are real because of the PT symmetry of the eigenvector ${\bf w}_0 = P \bar{\bf w}_0$
and the nonlinear field ${\bf N}_{\rm non}$ satisfies the second identity
(\ref{nonlinear-P}). Indeed, we have
$$
\langle P {\bf w}_0, {\bf w}_0 \rangle = \langle \bar{\bf w}_0, {\bf w}_0 \rangle =
\sum_{n = 1}^{2N} ({\bf w}_0)_n^2 = \sum_{n = 1}^{N} \left[ ({\bf w}_0)_n^2 + (\bar{\bf w}_0)_n^2 \right]
$$
and
$$
\langle P {\bf w}_0, {\bf N}_{\rm non}({\bf w}_0) \rangle = \langle \bar{\bf w}_0, {\bf N}_{\rm non}({\bf w}_0) \rangle =
\sum_{n = 1}^{2N} |({\bf w}_0)_n|^2 ({\bf w}_0)_n^2 = \sum_{n = 1}^{N} |({\bf w}_0)_n|^2
\left[ ({\bf w}_0)_n^2 + (\bar{\bf w}_0)_n^2 \right].
$$
Therefore, $\Delta$ is real at the leading order $\mathcal{O}(|a|^2)$. To exclude the gauge transformation,
let us consider the real values of $a$. Because the nonlinear vector field ${\bf N}_{\rm non}$
preserves the PT-symmetry, the unique solution for ${\bf u}$ and $\Delta$ is $PT$-symmetric,
so that ${\bf u} = P \bar{\bf u}$ and $\Delta$ is real. The bound (\ref{bound-local}) follows from
(\ref{decomposition-local}), (\ref{bound-on-u}), and (\ref{decomposition-final}). To be precise, we obtain
$$
\Delta = \Delta_2 a^2 + \mathcal{O}(a^4), \quad
\Delta_2 := \frac{\langle P {\bf w}_0, {\bf N}_{\rm non}({\bf w}_0) \rangle}{\langle P {\bf w}_0, {\bf w}_0 \rangle}.
$$
It remains to prove that $\Delta_2 > 0$. However, using the explicit representation from Theorem \ref{theorem-linear-PT},
for the eigenvalue with $\theta = \frac{2\pi j}{1 + 2N}$, $1 \leq j \leq N$, we obtain the eigenvector ${\bf w}_0$ with
components
$$
w_{2k-1} = \sqrt{E-i\gamma} \sin \frac{2\pi j (N+1-k)}{1+2N}, \quad w_{2k} = \sqrt{E+i\gamma} \sin \frac{2\pi j (N+1/2-k)}{1+2N},
\quad 1 \leq k \leq N.
$$
Therefore,
$$
\Delta_2 = \sqrt{E^2 + \gamma^2} \frac{\sum_{k=1}^N \sin^4 \frac{2\pi j (N+1-k)}{1+2N}}{\sum_{k=1}^N \sin^2 \frac{2\pi j (N+1-k)}{1+2N}} > 0,
$$
and the proof of the theorem is complete.
\end{proof}

\begin{remark}
The local bifurcation results do not apply in the limit $N \to \infty$ because of two reasons.
First, the spectrum of the linear stationary dNLS equation (\ref{dnls-linear}) becomes continuous
as $N \to \infty$. Second, for any $\gamma \neq 0$, the spectrum includes complex (purely imaginary)
points of $E$ because $\gamma_N \to 0$ as $N \to \infty$.
\end{remark}

Let us consider the simplest example $N = 1$
when Theorem \ref{theorem-local-bifurcation} works.
The two simple eigenvalues are
$E_0 = \pm \sqrt{1 - \gamma^2}$
and the eigenvectors ${\bf w}_0 = P \bar{\bf w}_0$ are given by the same expression
$$
{\bf w}_0 = \frac{\sqrt{3}}{2} \left[ \begin{array}{c} \sqrt{E_0 - i \gamma} \\ \sqrt{E_0 + i \gamma} \end{array} \right].
$$
In this case, $\Delta_2 = \frac{3}{4}$, so that we have the expansion
$$
E = \pm \sqrt{1 - \gamma^2} + A^2 + \mathcal{O}(A^4), \quad A := \frac{\sqrt{3}}{2} a.
$$
In fact, it follows from the exact solution (\ref{solution_N_1}) below
that the error term $\mathcal{O}(A^4)$ is identically zero.

\subsection{Stationary states: bifurcation from infinity}

We shall now consider the stationary states of the
nonlinear stationary PT-dNLS equation (\ref{dnls-nonlinear})
in the limit of large values of $E$. This corresponds
to the anti-continuum limit of weak couplings in the PT-dNLS lattice
after a suitable scaling transformation 
(which is also discussed in \cite{Pelin1}). Note
that the standard anti-continuum limit arising when the coupling
parameter in front of the discrete Laplacian operator vanishes
fails to generate any solutions of the stationary dNLS equation (\ref{dnls-nonlinear})
for real values of $E$ and $\gamma \neq 0$.

We shall develop methods to analyze a {\em bifurcation from infinity} for solution branches.
In particular, we shall prove the existence of $2^{N}$ branches of
the PT-symmetric stationary states ${\bf w}$ of the nonlinear stationary
PT-dNLS equation (\ref{dnls-nonlinear}) for $\gamma \in (-1,1)$
and for large values of $E$, for which $|w_n|^2$ is large for all $n \in S_N$.
The solution branches are unique up to the gauge transformation
${\bf w} \to e^{i\alpha} {\bf w}$ with $\alpha \in \R$.
The complication of proving this result is caused by the degeneracy
of asymptotic solutions of the nonlinear algebraic system (\ref{dnls-nonlinear})
in the limit $E \to \infty$. Indeed, setting ${\bf w} = \sqrt{E} {\bf W}$
and taking the limit $E \to \infty$, we obtain an uncoupled set of algebraic
equations with $N$ PT-symmetric solutions
$$
{\bf W}_k = e^{-i \varphi_k} {\bf e}_{k} + e^{i \varphi_k} {\bf e}_{2N+1-k}, \quad 1 \leq k \leq N,
$$
where $\varphi_k \in \R$ is an arbitrary parameter and ${\bf e}_k$ is a unit vector on the finite chain $S_N$.
However, the space of solutions of the nonlinear
algebraic system (\ref{dnls-nonlinear}) in the limit $E \to \infty$ does not enjoy the
linear superposition principle and parameters $\{ \varphi_k \}_{k = 1}^N$ must be fixed 
from $\mathcal{O}(1)$ conditions as $E \to \infty$.
To prove persistence of continuations of the limiting roots for large but finite values
of $E$, we have to unfold the degeneracy of the nonlinear system by a special transformation, after which the
result is guaranteed by the implicit function theorem. Along these lines, we prove 
the following main result.

\begin{theorem}
\label{theorem-nonlocal-bifurcation}
For any $\gamma \in (-1,1)$,
the nonlinear stationary PT-dNLS equation (\ref{dnls-nonlinear}) in the limit of large real $E$
admits $2^N$ PT-symmetric solutions ${\bf w} = P \bar{\bf w}$
(unique up to a gauge transformation) such that,
for sufficiently large $E$, the map $E \to {\bf w}$ is $C^{\infty}$ at each solution and
there is a positive $E$-independent constant $C$ such that
\begin{equation}
\label{bound-nonlocal}
\left| \sum_{n \in S_N} |w_n|^2 - 2 N E \right| \leq C.
\end{equation}
\end{theorem}

\begin{proof}
We set $E = \frac{1}{\delta}$ and ${\bf w} = \frac{{\bf W}}{\sqrt{\delta}}$
for small positive $\delta$ and write the stationary dNLS equation
(\ref{dnls-nonlinear}) in the equivalent form:
\begin{equation}
\label{dnls-nonlinear-equivalent}
(1 - |W_n|^2) W_n = \delta \left( W_{n+1} + W_{n-1} + i \gamma (-1)^n W_n \right), \quad n \in S_N,
\end{equation}
subject to the Dirichlet boundary conditions $W_0 = W_{2N+1} = 0$. We consider a PT-symmetric
solution with ${\bf W} = P \bar{\bf W}$ such that the system can be closed
at $N$ algebraic equations for $1 \leq n \leq N$ subject to the reflection boundary condition
$W_{N+1} = \bar{W}_N$. Note that parameter $\gamma \in \R$ is fixed.\\

{\bf Case $N = 1$:} In this case, we only have one nonlinear algebraic equation to solve:
\begin{equation}
\label{dnls-nonlinear-one}
(1 - |W_1|^2) W_1 = \delta \left[ \bar{W}_{1} - i \gamma W_1 \right].
\end{equation}
Setting $W_1 = A_1^{1/2} e^{i \varphi_1}$, we separate the real and imaginary parts
of equation (\ref{dnls-nonlinear-one}) as follows:
$$
A_1 = 1 - \delta \cos(2 \varphi_1), \quad - \sin(2 \varphi_1) - \gamma = 0.
$$
For any $\gamma \in (-1,1)$, there exist two solutions for $\varphi$ in $[0,\pi]$
from the second equation written as $\sin(2 \varphi_1) = -\gamma$. For each $\varphi$,
we have a unique solution of the first equation written as $A_1 = 1 \mp \delta \sqrt{1 - \gamma^2}$, 
from which we see that $A_1 = 1 + \mathcal{O}(\delta)$ as $\delta \to 0$.\\

{\bf Case $N \geq 2$:}
Let us now unfold the degeneracy of the nonlinear algebraic system (\ref{dnls-nonlinear-equivalent}) in
the limit $\delta \to 0$ by using the transformation:
\begin{equation}
\label{parameter-1}
\left\{ \begin{array}{l} W_1 = A_1^{1/2} e^{i \varphi_1}, \\
W_2 = (A_1 A_2)^{1/2} e^{i \varphi_1 + i \varphi_2}, \\
W_3 = (A_1 A_2 A_3)^{1/2} e^{i \varphi_1 + i \varphi_2 + i \varphi_3}, \\
\vdots \\
W_{N} = (A_1 A_2 \cdots A_N)^{1/2} e^{i \varphi_1 + i \varphi_2 + \cdots + i \varphi_N}, \end{array} \right.
\end{equation}
where amplitudes $A_1$,$A_2$,...,$A_N$ and phases $\varphi_1$, $\varphi_2$, ..., $\varphi_N$ are all real.
After substitution and separation of real and imaginary parts, we obtain $N$ equations
for phases
\begin{equation}
\left\{ \begin{array}{l} A_2^{1/2} \sin(\varphi_2) - \gamma = 0, \\
A_3^{1/2} \sin(\varphi_3) - A_2^{-1/2} \sin(\varphi_2) + \gamma = 0, \\
A_4^{1/2} \sin(\varphi_4) - A_3^{-1/2} \sin(\varphi_3) - \gamma = 0, \\
\vdots \\
- \sin 2(\varphi_1 + \varphi_2 + \cdots + \varphi_N) - A_{N}^{-1/2} \sin(\varphi_{N}) + (-1)^N \gamma = 0,  \end{array} \right.
\label{parameter-2}
\end{equation}
and $N$ equations for amplitudes
\begin{equation}
\label{parameter-3}
\left\{ \begin{array}{l} 1 - A_1 = \delta A_2^{1/2} \cos(\varphi_2), \\
1 - A_1 A_2 = \delta (A_3^{1/2} \cos(\varphi_3) + A_2^{-1/2} \cos(\varphi_2)), \\
1 - A_1 A_2 A_3 = \delta (A_4^{1/2} \cos(\varphi_4) + A_3^{-1/2} \cos(\varphi_3)), \\
\vdots \\
1 - A_1 A_2 \cdots A_N = \delta (\cos 2(\varphi_1 + \varphi_2 + \cdots + \varphi_N) + A_{N}^{-1/2} \cos(\varphi_{N})).
\end{array} \right.
\end{equation}

For $\delta = 0$, the system of amplitude equations (\ref{parameter-3}) has a unique solution
at the point $A_1 = A_2 = \cdots = A_N = 1$. The vector field of the nonlinear system is smooth
with respect to $(A_1,A_2,\ldots,A_N)$ and $\delta$ near this point
for all $(\varphi_1,\varphi_2,\ldots,\varphi_N) \in \mathbb{T}^N$,
where $\mathbb{T}$ denotes the fundamental interval $[0,2\pi]$ subject to the periodic boundary conditions.
The Jacobian matrix with respect to $(A_1,A_2,\ldots,A_N)$
at this point has eigenvalue $1$ of geometric multiplicity one and algebraic multiplicity $N$.
By the Implicit Function Theorem, for all $(\varphi_1,\varphi_2,\ldots,\varphi_N) \in \mathbb{T}^N$
and small $\delta \in \R$, there is a unique solution of the nonlinear system (\ref{parameter-3})
such that the map $(\varphi_1,\varphi_2,\ldots,\varphi_N,\delta) \to (A_1,A_2,\ldots,A_N)$
is $C^{\infty}$ and there is a positive $\delta$-independent constant $C$ such that
\begin{equation}
\label{bound-1}
|A_1 - 1| + |A_2 - 1| + \cdots + |A_N - 1| \leq C |\delta|.
\end{equation}
Bound (\ref{bound-nonlocal}) follows from this bound and the scaling transformation. 

Now we consider the system of phase equations (\ref{parameter-2}),
which is $\delta$ independent. Nevertheless, it depends on $\delta$ via amplitudes
$(A_1,A_2,\ldots,A_N)$. For $\delta = 0$, the nonlinear system (\ref{parameter-2})
can be written in the explicit form:
\begin{equation}
\left\{ \begin{array}{l} \sin(\varphi_2) = \gamma, \\
\sin(\varphi_3) = - \gamma + \sin(\varphi_2) \equiv 0, \\
\sin(\varphi_4) = \gamma + \sin(\varphi_3) \equiv \gamma, \\
\vdots \\
\sin 2(\varphi_1 + \varphi_2 + \cdots + \varphi_N) = (-1)^N \gamma - \sin(\varphi_{N}).  \end{array} \right.
\label{parameter-4}
\end{equation}
Denote $\psi := 2(\varphi_1 + \varphi_2 + \cdots + \varphi_N)$.
For any $\gamma \in (-1,1)$, there are $2^N$ possible solutions
of (\ref{parameter-4}) for $(\psi,\varphi_2,\ldots,\varphi_N) \in \mathbb{T}^N$,
depending on the binary choice of the roots of the sinusoidal functions on the fundamental period.
Because $\varphi_1 = \frac{\psi}{2} - \varphi_2 - \cdots - \varphi_N$, there are
actually four solutions for $\varphi_1$ in $\mathbb{T}$, however, the solutions
with $\varphi_1 \in (\pi,2\pi]$ are reducible to the solutions with $\varphi_1 \in (0,\pi]$
by the transformation ${\bf W} \to -{\bf W}$, which is a particular case of
the gauge transformation. In what follows, we only consider the two possible
solutions for $\varphi_1$ in $[0,\pi]$.

The vector field of the nonlinear system (\ref{parameter-2}) with  $(A_1,A_2,\ldots,A_N)$
obtained from the nonlinear system (\ref{parameter-3}) is smooth
with respect to $(\varphi_1,\varphi_2,\ldots,\varphi_N)$ and $\delta$.
The Jacobian matrix with respect to $(\varphi_1,\varphi_2,\ldots,\varphi_N)$
for $\delta = 0$ is given by the matrix
$$
\left[ \begin{array}{cccccc} 0 & \cos(\varphi_2) & 0 & 0 & \cdots & 0 \\
0 & -\cos(\varphi_2) & \cos(\varphi_3) & 0 & \cdots & 0 \\
0 & 0 & -\cos(\varphi_3) & \cos(\varphi_4) & \cdots & 0 \\
\vdots & \vdots & \vdots & \cdots & \ddots & \vdots \\
-2 \cos(\psi) & -2 \cos(\psi)
& -2 \cos(\psi) & -2 \cos(\psi) & \cdots & -\cos(\varphi_N)
\end{array} \right],
$$
Now it is clear that $\cos(\varphi_n) \neq 0$ for all $2 \leq n \leq N$
if $\gamma \in (-1,1)$. In addition, the last equation
in the system (\ref{parameter-4}) is given by either $\sin(\psi) = 0$
if $N$ is even or $\sin(\psi) = \gamma$ if $N$ is odd. In either case,
$\cos(\psi) \neq 0$ if $\gamma \in (-1,1)$. Hence, the Jacobian matrix
is invertible if $\gamma \in (-1,1)$. By the Implicit Function Theorem,
for all small $\delta \in \R$, there is a unique continuation of
any of the $2^N$ possible solutions $(\varphi_1^*,\varphi_2^*,\ldots,\varphi_N^*)$
of the nonlinear system (\ref{parameter-4})
as a solution of the nonlinear system (\ref{parameter-2})
such that the map $\delta \to (\varphi_1,\varphi_2,\ldots,\varphi_N)$
is $C^{\infty}$ and there is a positive $\delta$-independent constant $C$ such that
\begin{equation}
\label{bound-2}
|\varphi_1 - \varphi_1^*| + |\varphi_2 - \varphi_2^*| + \cdots + |\varphi_N - \varphi_N^*| \leq C |\delta|.
\end{equation}
This completes the proof of the theorem.
\end{proof}

\begin{remark}
The number of solution branches grows as $N \to \infty$ for any fixed value of $\gamma$ in
the interval $(-1,1)$. However, all these solution branches are delocalized
in the sense that $|w_n|^2 \approx E$ as $E \to \infty$ for all $n$ in $S_N$. Therefore,
none of the solution branches of Theorem \ref{theorem-nonlocal-bifurcation}
approach to a localized state (discrete soliton) as $N \to \infty$.
\end{remark}

\begin{remark}
Besides solution branches of Theorem \ref{theorem-nonlocal-bifurcation},
for any $N \geq 2$ and $1 \leq M \leq N$, there exist additional solution
branches such that $|w_n|^2 \approx E$ as $E \to \infty$ for $N-M+1 \leq n \leq N+M$
and $|w_n|^2 \approx 0$ as $E \to \infty$ for $1 \leq n \leq N-M$ and $N+M+1 \leq n \leq 2N$.
These stationary states are supported at $2M$ sites near the central sites in $S_N$
and their persistence is proved with a similar variant of the implicit function theorem
(see the proof of Theorem \ref{theorem-soliton} below).
If $N \to \infty$, such stationary states approach to a localized state (discrete soliton). 
Note that the discrete solitons are unstable on the unbounded lattice because the
continuous spectrum of the linearized dNLS equation (\ref{dnls-linear}) is complex for any $\gamma \neq 0$, recall that
$\gamma_N \to 0$ as $N \to \infty$. \label{remark-other-branches}
\end{remark}

\begin{remark}
The arguments of the implicit function theorem can not be applied to construct 
solution branches which are centered anywhere but at the central sites in the set $S_N$.
Indeed, the numerical results below show that no such solution branches exist for large
values of $E$.
\label{remark-numerical-branches}
\end{remark}

\subsection{Numerical results}

We shall construct here the simplest nonlinear stationary states for $N = 1, 2, 3$.
For $N = 1$, this corresponds to the nonlinear dimer, where the solution branches can
be obtained analytically, as in \cite{Li,Ramezani,Sukhorukov}. For $N = 2$, this corresponds
to the nonlinear quadrimer and the solution branches can be at best approximated numerically
\cite{Li,ZK}. For $N = 3$, the numerical approximations of the nonlinear stationary states
are added here for the first time.

For $N = 1$, we use the reduction $w_2 = \bar{w}_1$ and write $w_1 = A e^{-i \varphi}$
with real $A$ and $\varphi$. Then, the stationary PT-dNLS equation (\ref{dnls-nonlinear}) yields
two equations
\begin{equation}
\label{N_1}
\sin(2 \varphi) = \gamma, \quad A^2 = E - \cos(2 \varphi).
\end{equation}
With two solutions of the first equation for $\varphi \in [0,\pi]$, we obtain two solution branches
\begin{equation}
\label{solution_N_1}
A^2_{\pm} = E \mp \sqrt{1 - \gamma^2}, \quad \gamma \in (-\gamma_1,\gamma_1),
\end{equation}
where $\gamma_1 = 1$. The two solution branches coalesce into one branch for
$\gamma = \gamma_1$ and disappear via a saddle-center bifurcation for $\gamma > \gamma_1$.

Positivity of $A^2_{\pm}$ shows that $E > E_{\pm} = \pm \sqrt{1 - \gamma^2}$,
where $E_{\pm}$ are the simple eigenvalues of the linear stationary PT-dNLS equation for $\gamma \in (-\gamma_1,\gamma_1)$.
We note that $A^2_{\pm} \to 0$ as $E \to E_{\pm}$ and that $A^2_{\pm} \sim E$ as $E \to \infty$.
These analytical results clearly illustrate the bifurcation results in Theorems
\ref{theorem-local-bifurcation} and \ref{theorem-nonlocal-bifurcation}.

\begin{figure}[tbp]
\begin{center}
\includegraphics[width=50mm,height=40mm]{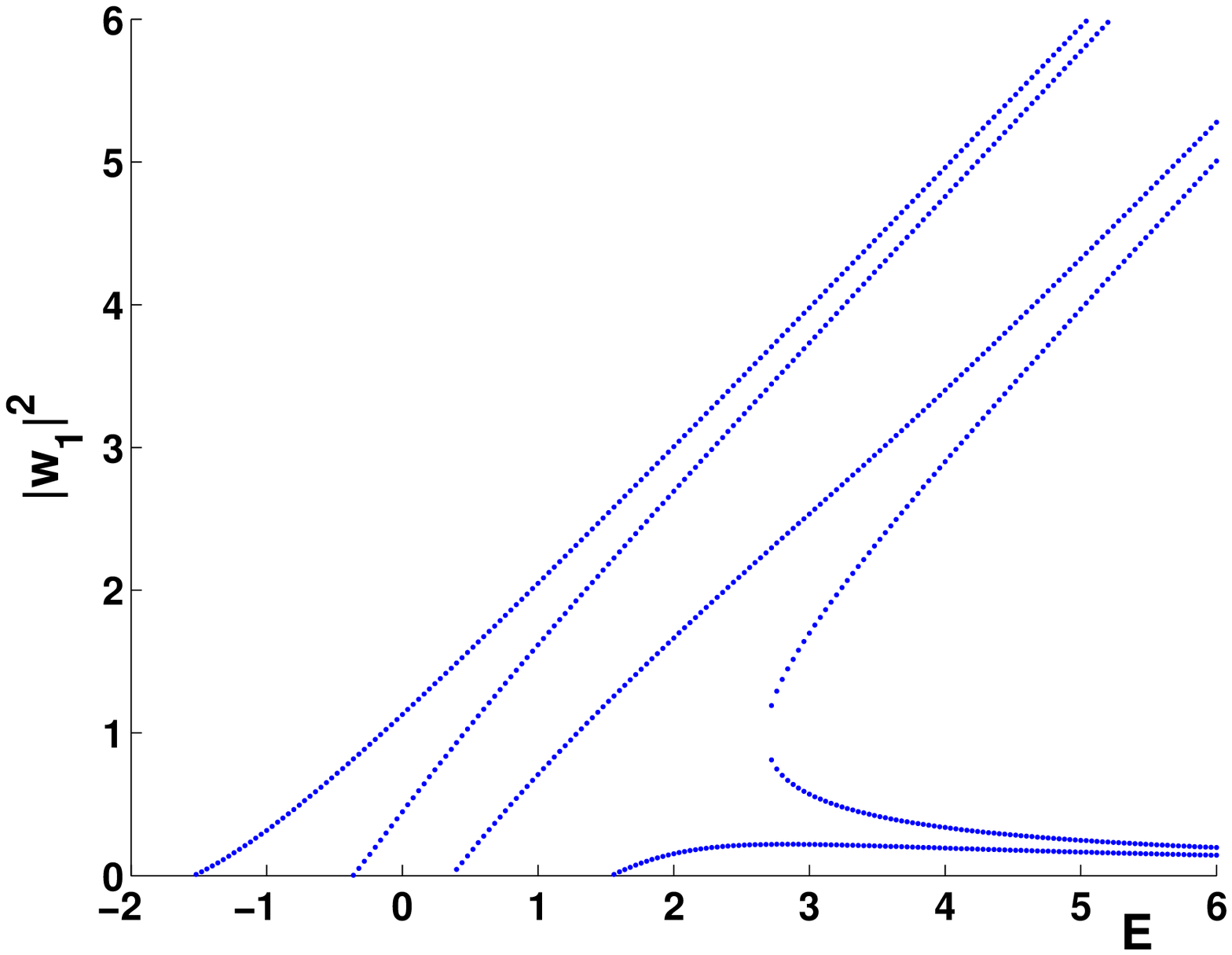}
\includegraphics[width=50mm,height=40mm]{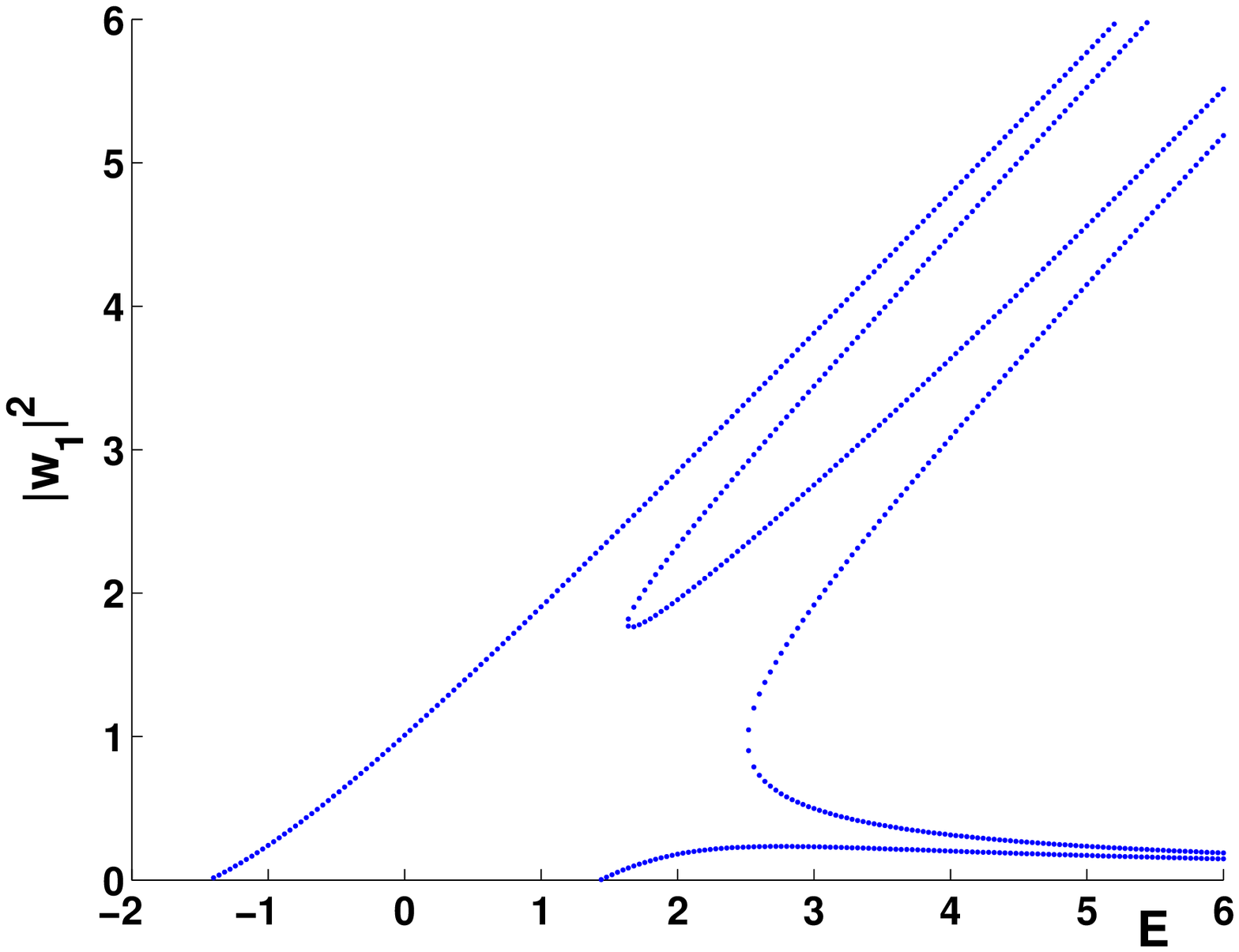}
\includegraphics[width=50mm,height=40mm]{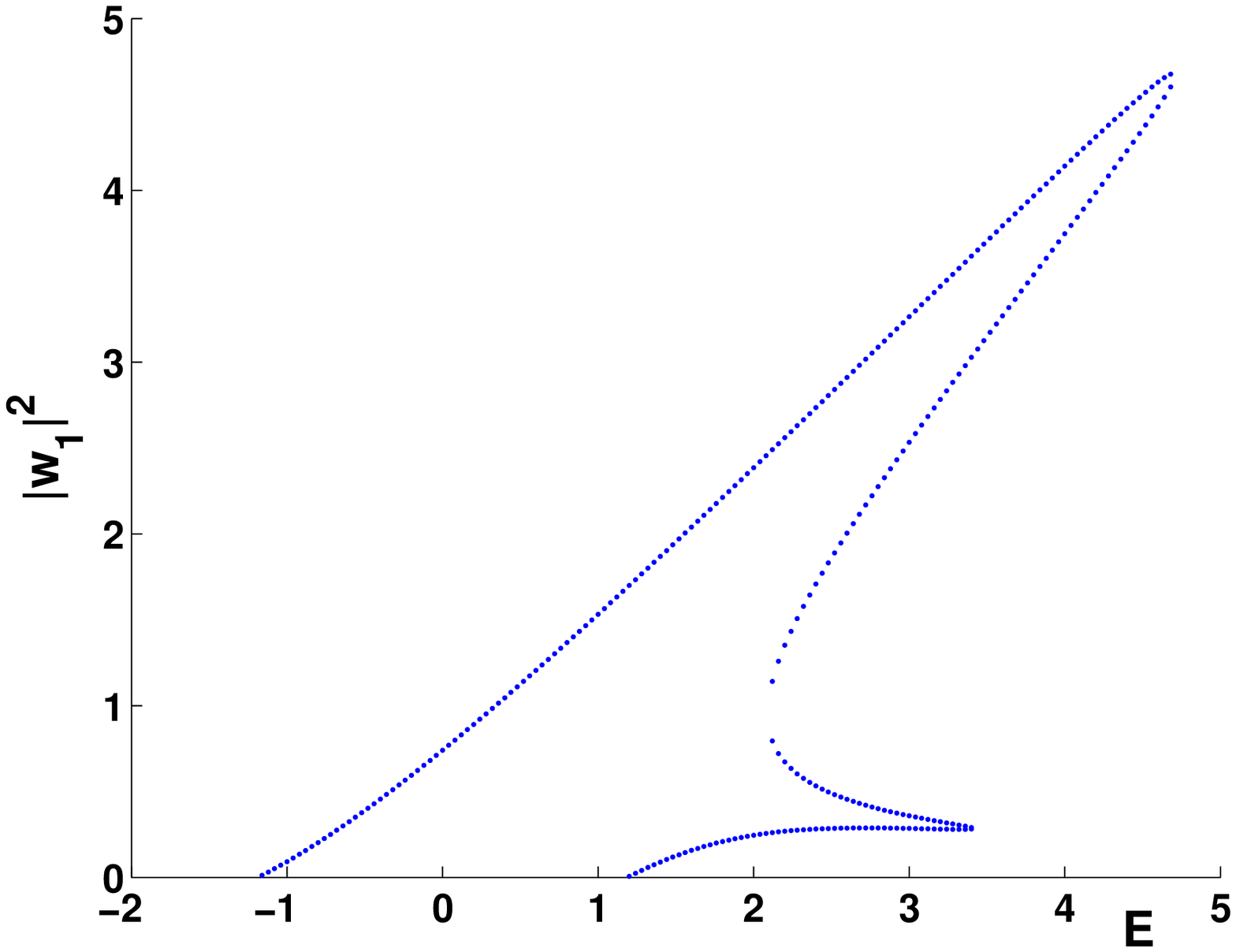} \\
\includegraphics[width=50mm,height=40mm]{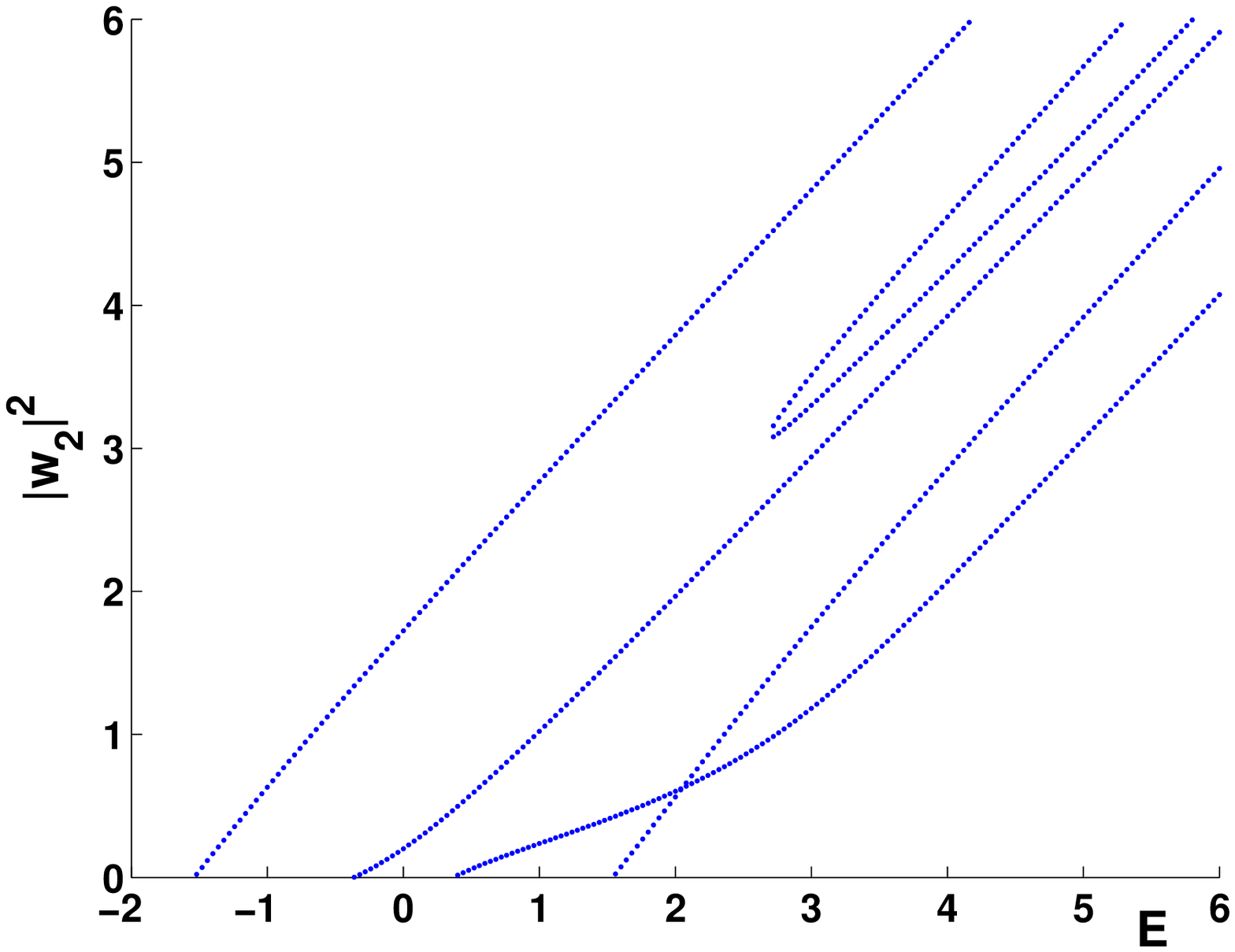}
\includegraphics[width=50mm,height=40mm]{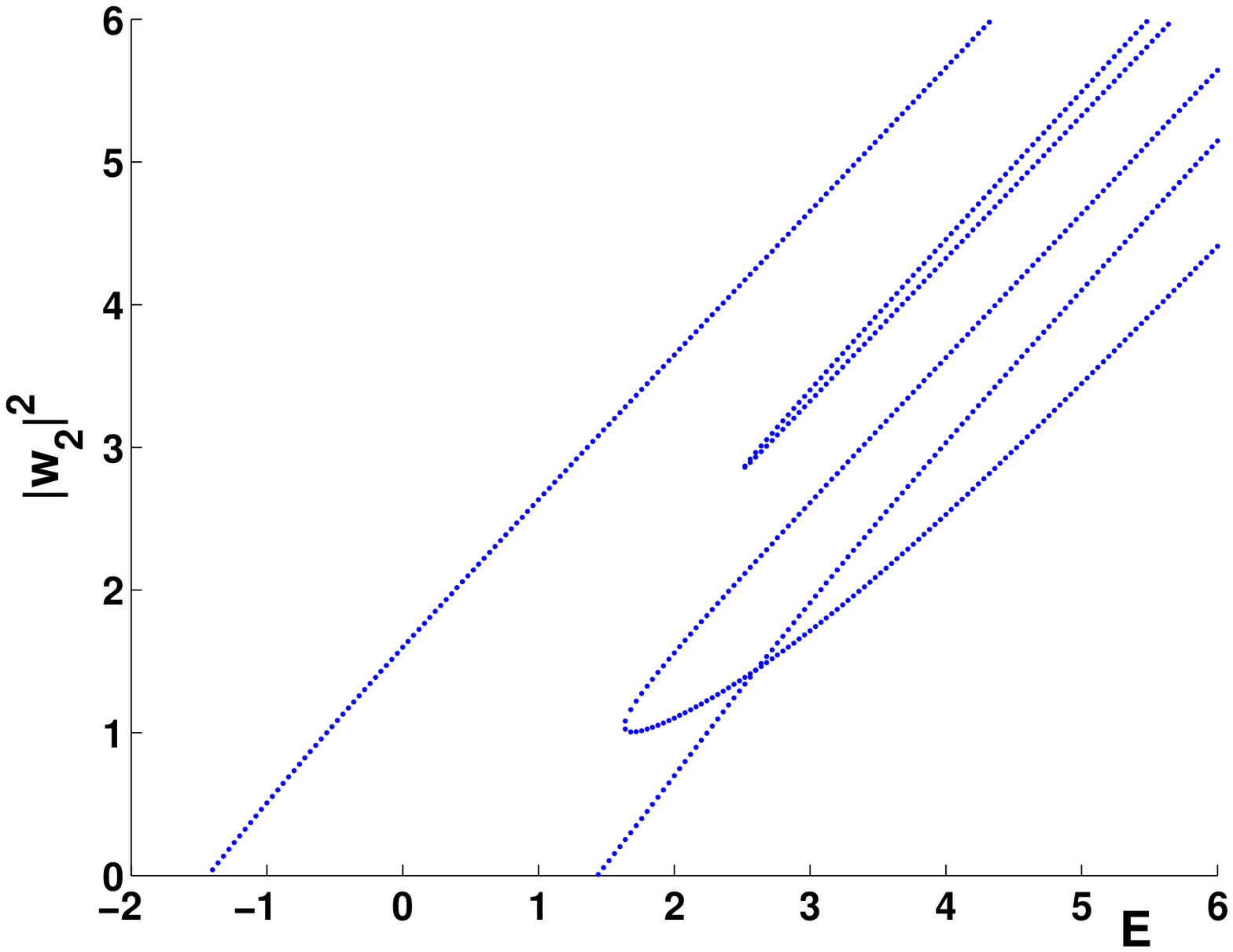}
\includegraphics[width=50mm,height=40mm]{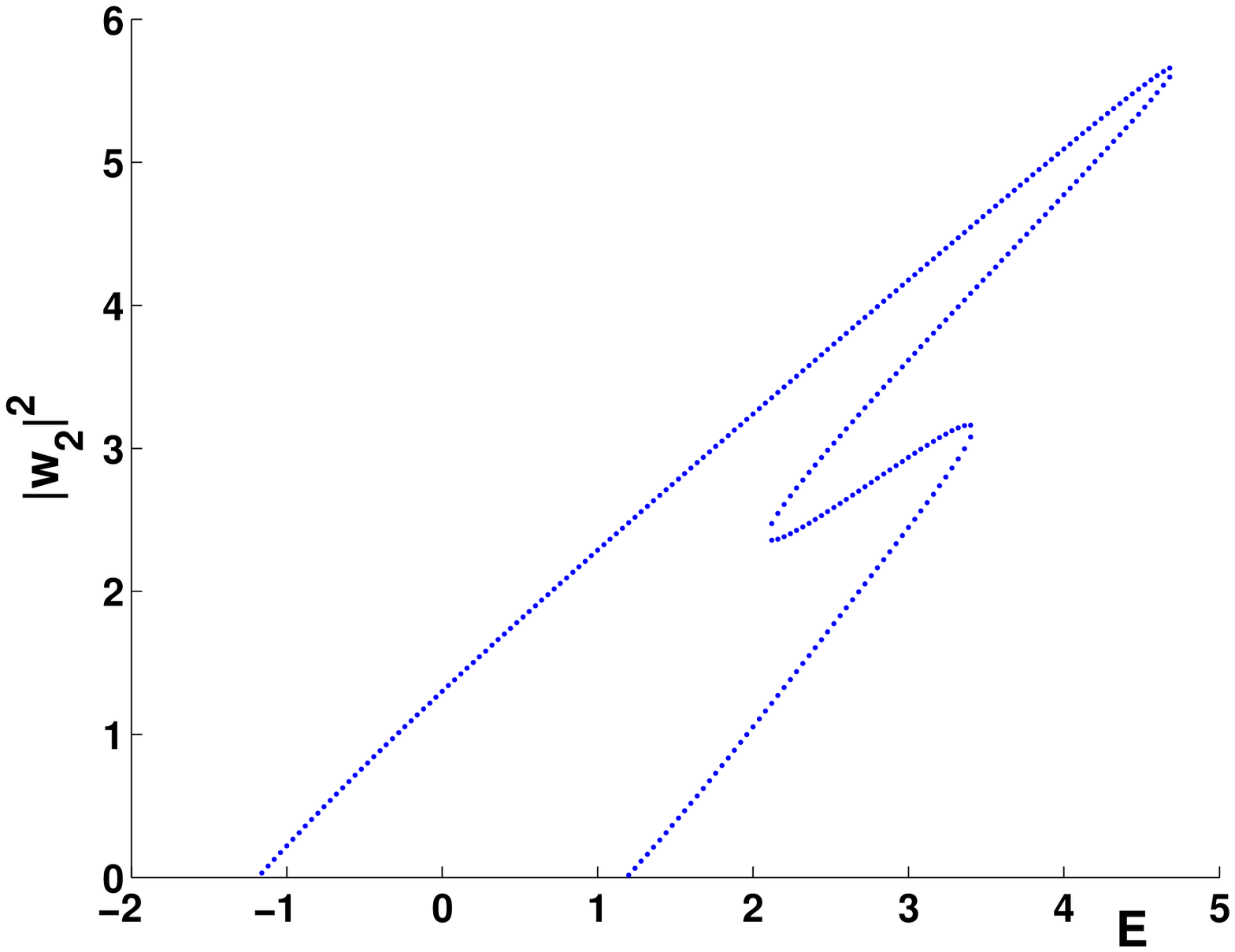}
\end{center}
\caption{Nonlinear stationary states for $N = 2$ and for
$\gamma = 0.5$ (left), $\gamma = 0.75$ (middle), and $\gamma = 1.1$ (right).
The top and bottom rows show components $|w_1|^2$ and $|w_2|^2$ respectively.}
\label{fig1}
\end{figure}

\begin{figure}[tbp]
\begin{center}
\includegraphics[width=50mm,height=40mm]{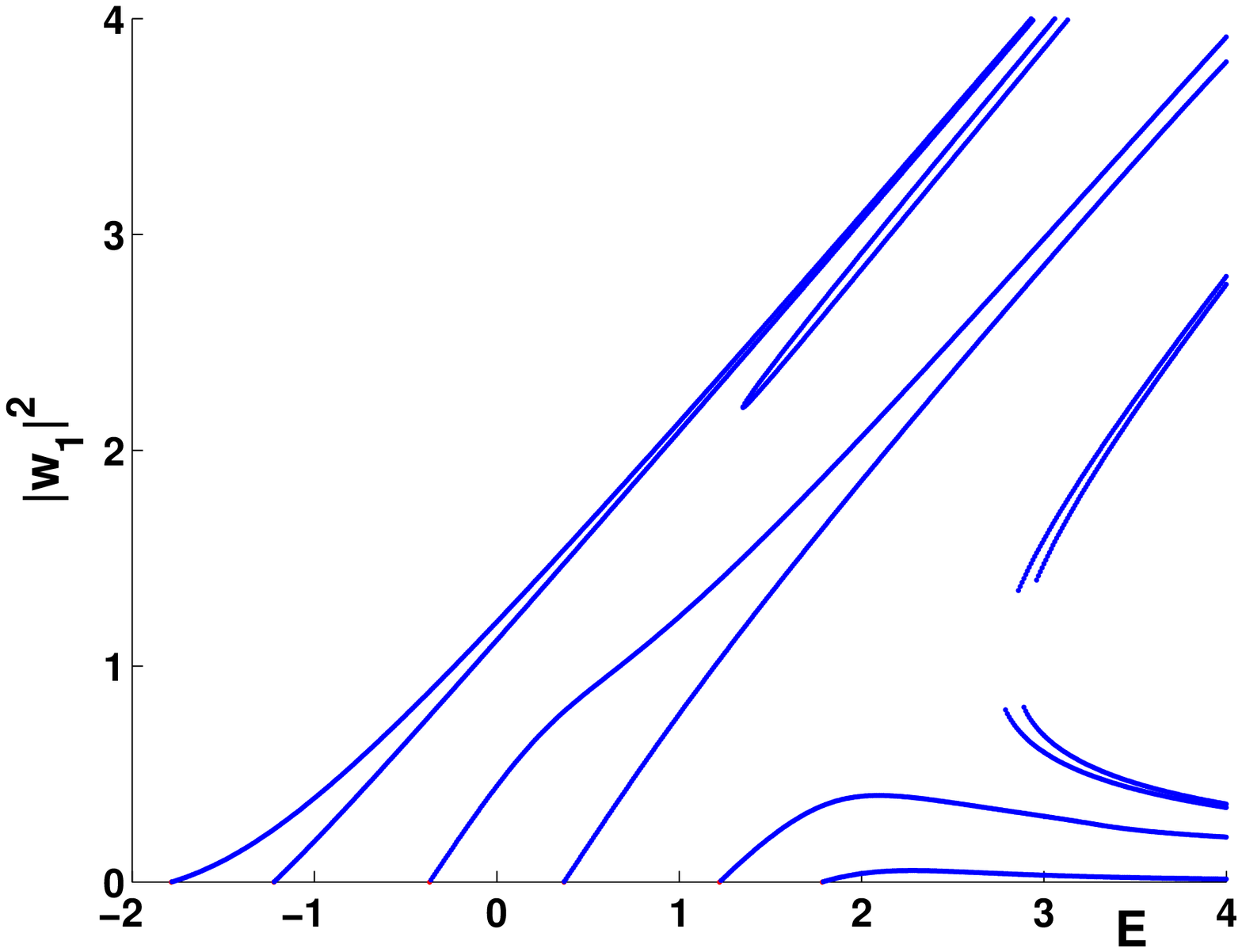}
\includegraphics[width=50mm,height=40mm]{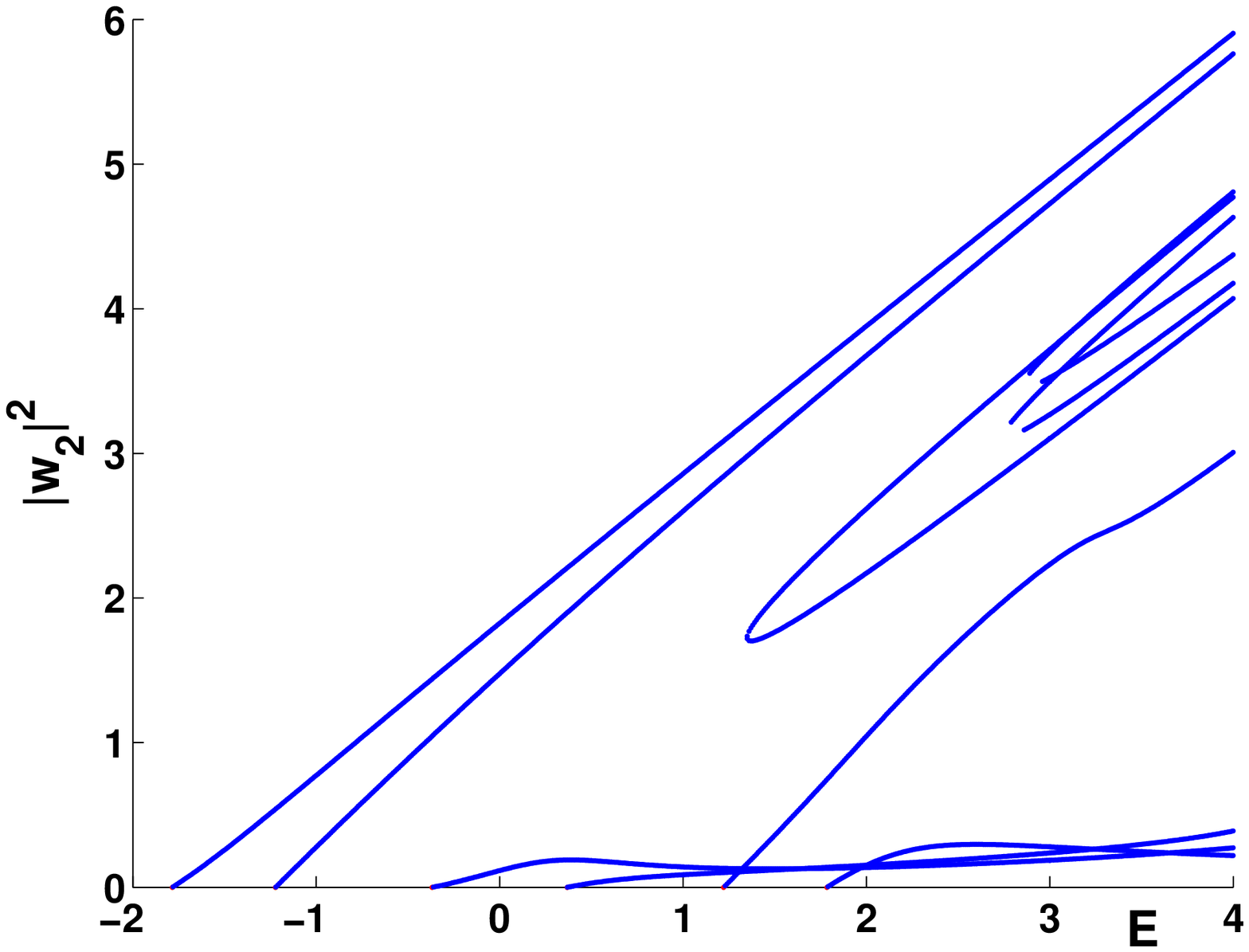}
\includegraphics[width=50mm,height=40mm]{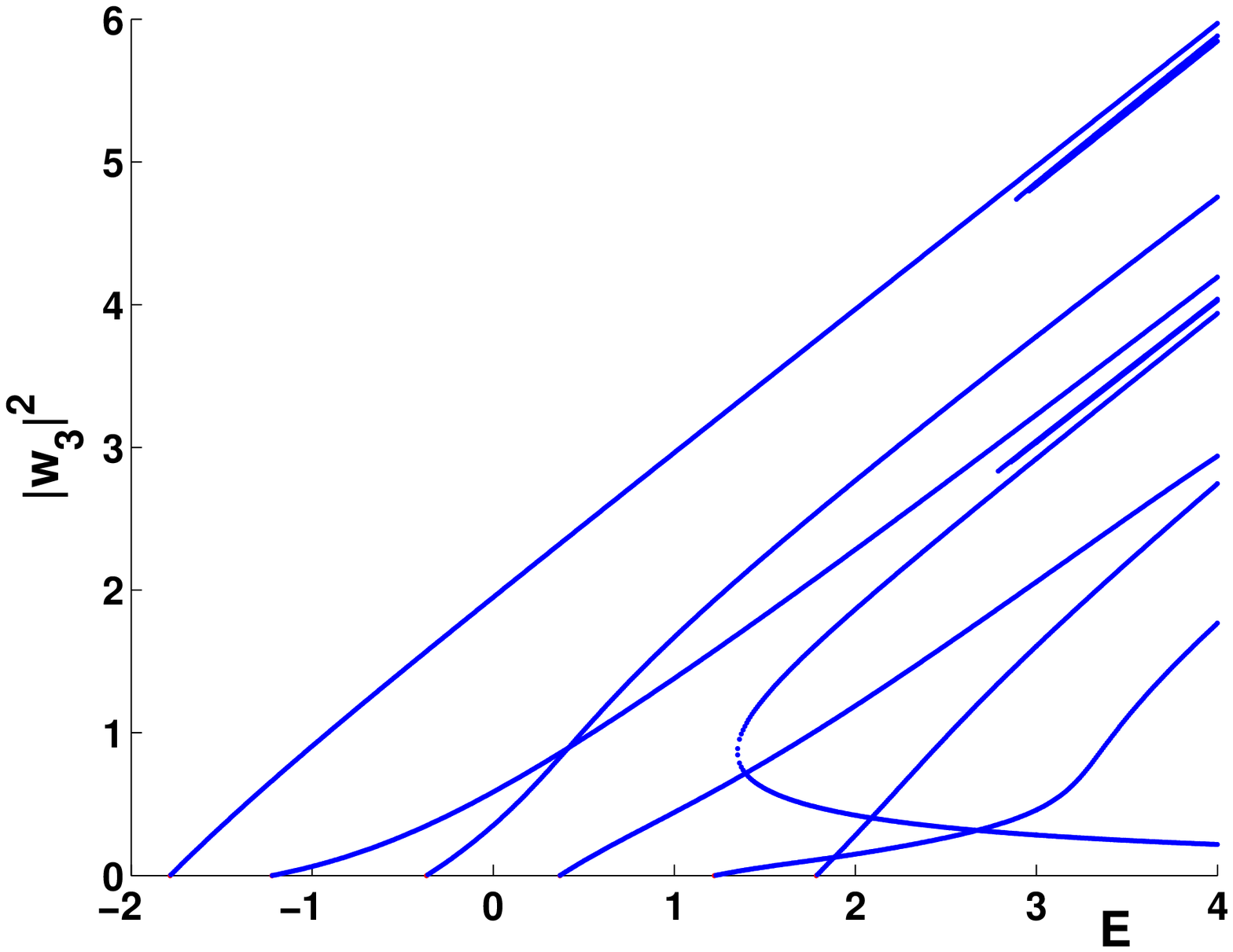} \\
\includegraphics[width=50mm,height=40mm]{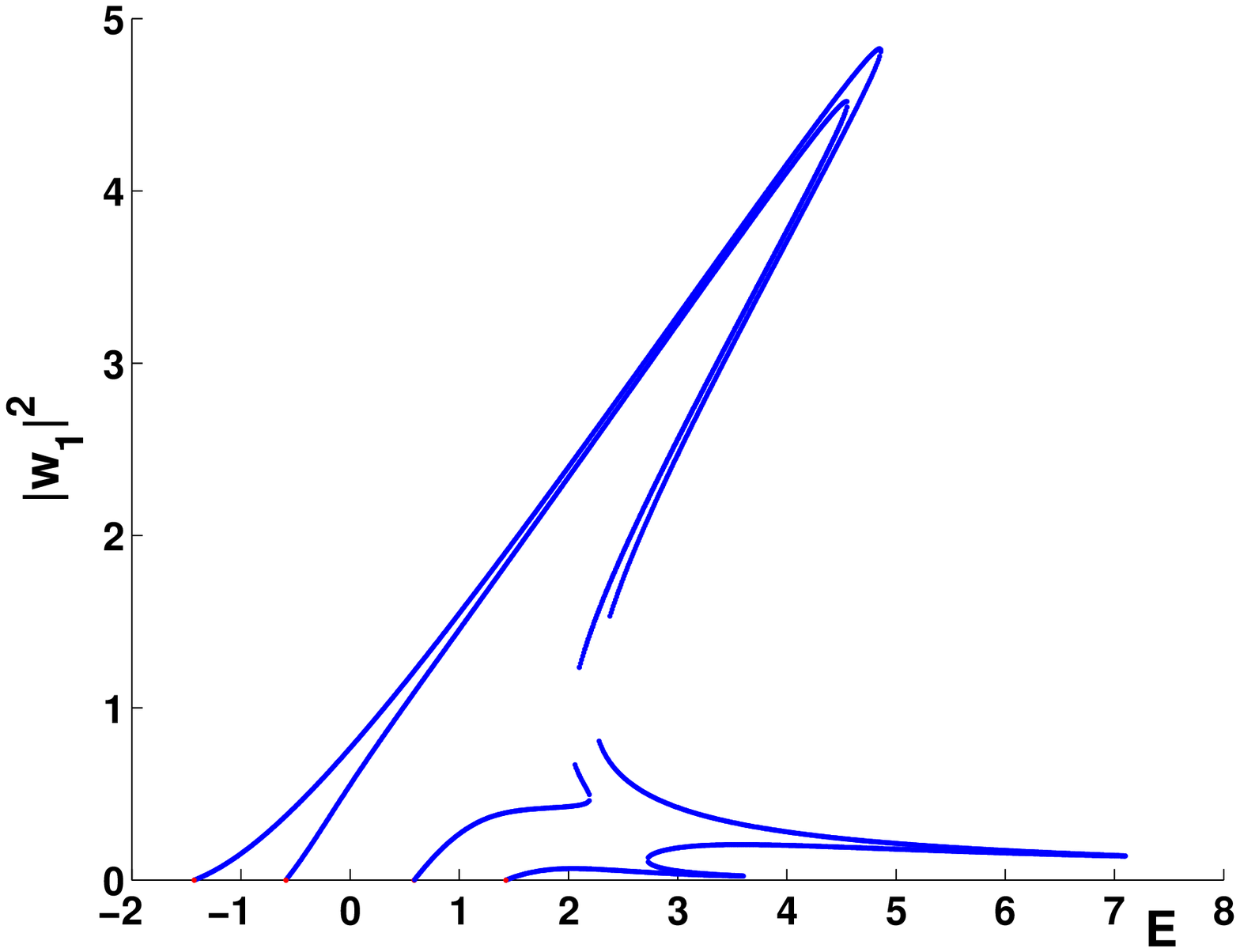}
\includegraphics[width=50mm,height=40mm]{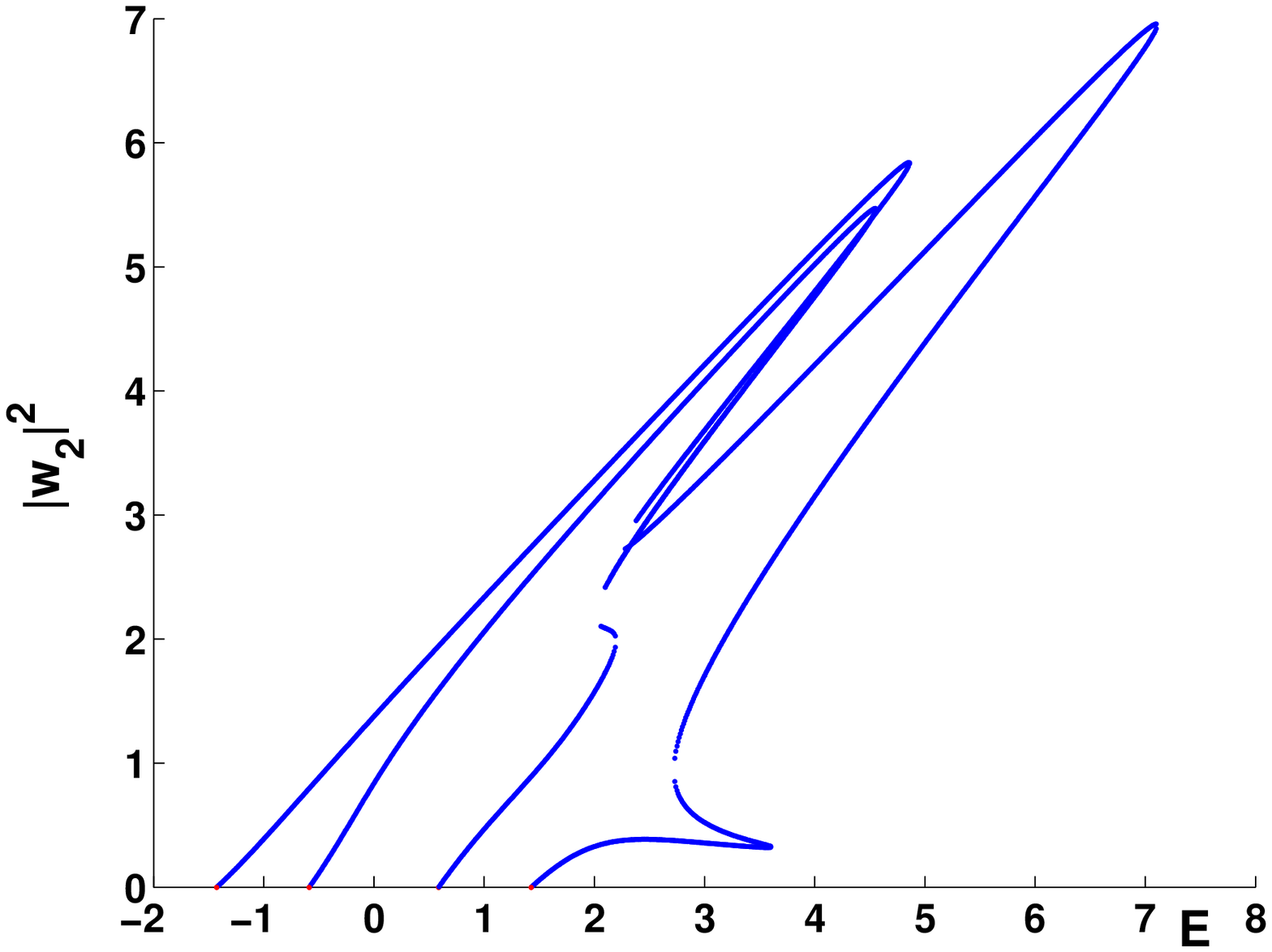}
\includegraphics[width=50mm,height=40mm]{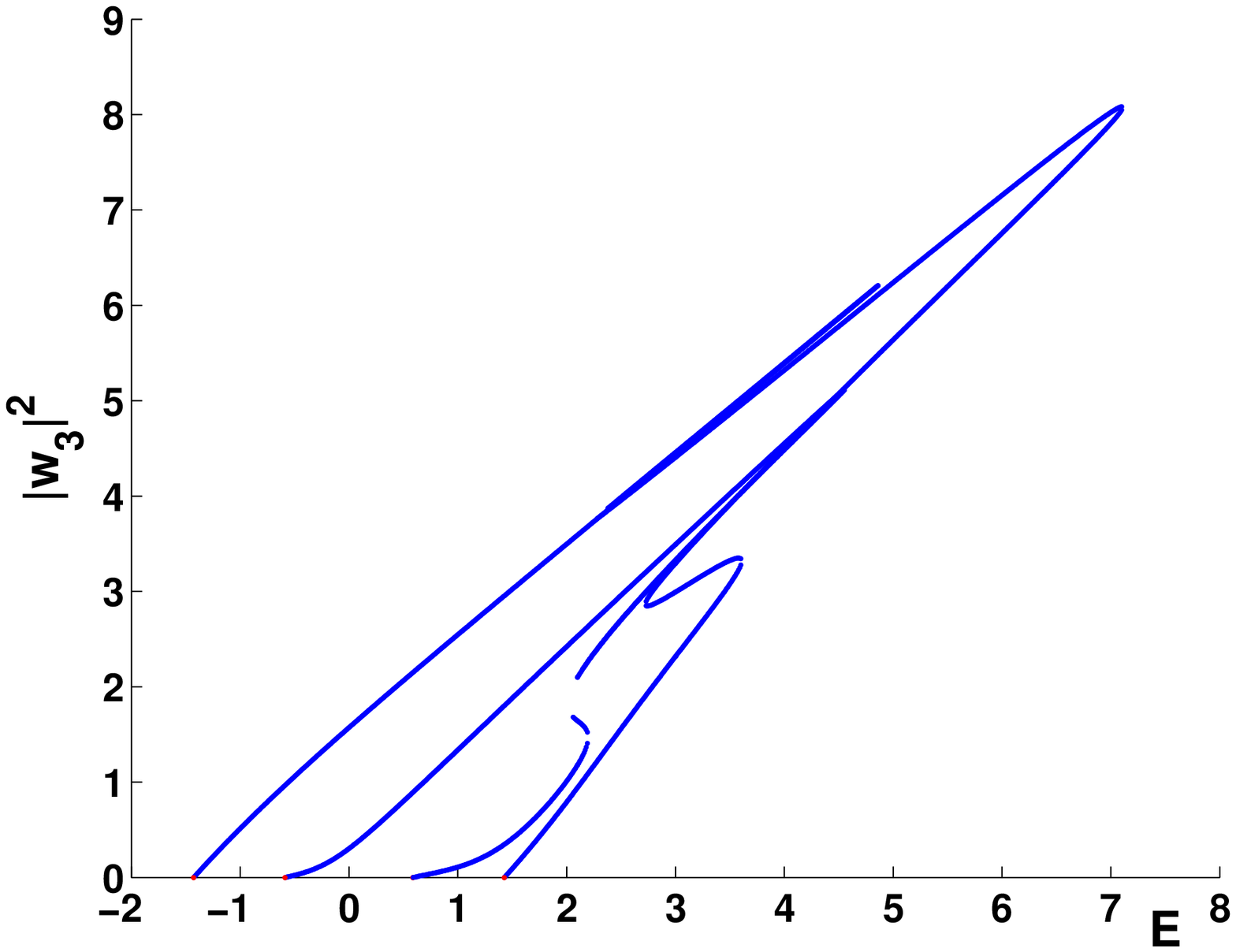}
\end{center}
\caption{Nonlinear stationary states for $N = 3$, and for
$\gamma = 0.25$ (top) and $\gamma = 1.1$ (bottom). The left, middle, and right columns
show components $|w_1|^2$, $|w_2|^2$, and $|w_3|^2$ respectively.}
\label{fig2}
\end{figure}

For $N = 2$, we use the reduction $w_4 = \bar{w}_1$ and $w_3 = \bar{w}_2$.
Writing $w_1 = A e^{-i \varphi - i \psi}$ and $w_2 = B e^{-i \psi}$
with real $A$, $B$, $\varphi$, and $\psi$ in the nonlinear stationary dNLS equation
(\ref{dnls-nonlinear}), we obtain the system of nonlinear equations:
\begin{equation}
\label{N_2}
\sin(\varphi) = \frac{\gamma A}{B}, \quad \sin(2 \psi) = \frac{\gamma (A^2-B^2)}{B^2}
\end{equation}
and
\begin{equation}
\label{N_2-amplitudes}
A^3 = A E - B \cos(\varphi), \quad B^3 = B E - A \cos(\varphi) - B \cos(2 \psi).
\end{equation}
The roots of the algebraic system (\ref{N_2}) and (\ref{N_2-amplitudes}) can be investigated numerically
and the results depend on the value of $\gamma$.
Figure 1 shows the solution branches on the $(E,A^2)$-plane (top) and the $(E,B^2)$-plane (bottom)
for $\gamma = 0.5 < \gamma_2 \approx 0.618$ (left),
$\gamma = 0.75$ (middle), and $\gamma = 1.1$. Note that no solution branches exist
for $\gamma > \gamma_2^* := 2 \cos \frac{\pi}{5} \approx 1.618$, because no simple real eigenvalues
occur in the linearized dNLS equation (\ref{dnls-linear}) for these values of $\gamma$.

According to Theorem \ref{theorem-local-bifurcation}, we count exactly four ($2N=4$) solution
branches for small amplitudes $A$ and $B$ for $\gamma < \gamma_2$
and exactly two small solution branches for $\gamma_2 < \gamma < \gamma_2^*$.
According to Theorem \ref{theorem-nonlocal-bifurcation}, we count exactly
four ($2^{N=2} = 4$) solution branches for large amplitudes $A$ and $B$ if $\gamma < 1$,
whereas all solution branches terminate before reaching large amplitudes
if $\gamma > 1$.

The two branches for small values of $A$ and large values of $E$ are attributed to the solutions
in Remark \ref{remark-other-branches}. The corresponding values of $B$ are large. On
the other hand, no branches exist for large $A$ and small $B$ as $E$ gets large,
see Remark \ref{remark-numerical-branches}.

For $N = 3$, we write $w_1 = A e^{-i(\varphi + \psi + \theta)} = \bar{w}_6$,
$w_2 = B e^{-i(\psi + \theta)} = \bar{w}_5$, and $w_3 = C e^{-i \theta} = \bar{w}_4$.
The roots of the resulting algebraic system are investigated numerically
by a homotopy method and the results are shown on Figure \ref{fig2}.
We count six ($2N = 6$) branches in the small-amplitude limit if $\gamma < \gamma_3 \approx 0.445$
and four branches if $\gamma_3 < \gamma < \gamma^*_3 ; =2 \cos \frac{2\pi}{7} \approx 1.247$.
We also count eight ($2^{N=3} = 8$) branches
for large amplitudes $A^2$ if $\gamma < 1$ and no branches for large amplitudes $A^2$
if $\gamma > 1$. More branches are counted for large amplitudes $B^2$ and even more branches
for large amplitudes $C^2$. Overall, the results for $N = 3$ are similar to the
results for $N = 2$.

\section{PT-symmetric defects embedded in infinite PT-dNLS lattices}

We shall now consider an infinite PT-dNLS
lattice, where the particular emphasis is
on the existence and stability of localized stationary states (discrete solitons).
Because the phase transition in the PT-dNLS equation (\ref{dnls}) on the infinite lattice occurs already at
$\gamma_{N \to \infty} = 0$, there is no way to obtain stable discrete solitons
in such systems with extended gain and loss \cite{Pelin2}. Therefore, we modify the PT-dNLS lattice by considering
the PT-symmetric potential as a finite-size defect. Such defects were
considered recently in the physical literature  \cite{Sukh} and \cite{lepri,Ambroise}.

Let $N$ be a positive integer and $S_N := \{ 1,2,...,2N\}$ be the sites of
the lattice, where the PT-symmetric defects are placed. The model takes the form
\begin{eqnarray}
i \frac{d u_n}{dt} = u_{n+1} - 2 u_n + u_{n-1}
+ i \gamma (-1)^n \chi_{n \in S_N} u_n + |u_n|^2 u_n, \quad n \in \Z,
\label{eqn0}
\end{eqnarray}
where $\chi_{n \in S_N}$ is a characteristic function for the set $S_N$.
When $N = 1$, the PT-dNLS
equation (\ref{eqn0}) corresponds to the embedded dimer
in the infinite PT-dNLS lattice. When $N = 2$, it corresponds to the embedded quadrimer,
and so on.

We study the linearized PT-dNLS equation and find the phase transition threshold $\tilde{\gamma}_N$.
It is quite remarkable that $\tilde{\gamma}_N > \gamma_N$ for any $N \in \N$, in
particular, $\tilde{\gamma}_1 = \sqrt{2}$. Nevertheless, $\tilde{\gamma}_N$ is still
a monotonically decreasing sequence of $N$ such that
$\tilde{\gamma}_N \to 0$ as $N \to \infty$.

Then, we employ the large-amplitude (anti-continuum) limit of the PT-dNLS equation (\ref{eqn0}) to study
the existence of discrete solitons supported at the PT-symmetric defect $S_N$.
For recent results on existence of discrete solitons in the anti-continuum limit for the regular dNLS equation
(in the absence of PT-symmetry), see e.g. \cite{Alfimov,CPS}.
We find that for all $\gamma \in (-1,1)$, $2^N$ branches of the discrete solitons
exist in this limit, for which $|u_n|^2$ is large for all $n \in S_N$.

The existence and stability of discrete solitons is illustrated numerically and
we show that the stable branches of the discrete solitons for $\gamma \neq 0$ 
originate from the stable branches in the Hamiltonian version ($\gamma = 0$)
of the dNLS equation \cite{pkf,PelSak}.

\subsection{Eigenvalues of the linear PT-dNLS equation}

We consider the linear stationary PT-dNLS equation:
\begin{equation}
\label{dnls-linear-defect}
E w_n = w_{n+1} + w_{n-1} + i \gamma (-1)^n \chi_{n \in S_N} w_n, \quad n \in \Z.
\end{equation}
Because the PT-symmetric potential is compact, the continuous spectrum of the linear PT-dNLS equation
(\ref{dnls-linear-defect}) is located for $E \in [-2, 2]$. Besides the continuous spectrum, isolated eigenvalues
may exist outside the continuous spectrum. To characterize isolated eigenvalues, we introduce a parametrization
\begin{equation}
\label{theta}
E := 2 \cos \theta, \quad {\rm Re}(\theta) \in [-\pi,\pi], \quad {\rm Im}(\theta) > 0,
\end{equation}
and look for exponentially decaying solutions of the linear PT-dNLS equation (\ref{dnls-linear-defect}) in the form:
\begin{equation}
w_n = \left\{ \begin{array}{lc} w_1 e^{-i \theta (n-1)}, \quad & n \leq 1, \\
w_{2N} e^{i \theta (n-2N)}, \quad & \;\; n \geq 2N, \end{array} \right.
\end{equation}
which still leaves a set of $2N$ unknown variables $\{ w_n \}_{n \in S_N}$.
To find $\{ w_n \}_{n \in S_N}$, we close the linear eigenvalue problem at the
algebraic system
\begin{equation}
\label{dnls-linear-closed}
2 \cos \theta w_n = w_{n+1} + w_{n-1} + i \gamma (-1)^n w_n, \quad n \in S_N,
\end{equation}
subject to the boundary conditions
$$
w_0 = w_1 e^{i \theta}, \quad w_{2N+1} = w_{2N} e^{i \theta}.
$$

Note that each eigenvalue $E$ is complex if ${\rm Im}(\theta) > 0$ and 
there exists a complex conjugate eigenvalue $\bar{E}$ by the PT-symmetry 
(see also Remark \ref{remark-eigenvalue}).
The following result is similar to the result of Theorem \ref{theorem-linear-PT}.

\begin{theorem}
A new symmetric pair of complex-conjugate eigenvalues 
of the linear PT-dNLS equation (\ref{dnls-linear-defect}) bifurcates
at $|\gamma| = \gamma_{N,k}$, where
\begin{equation}
\label{gamma-N-k}
\gamma_{N,k} := 2 \cos\frac{\pi (2k-1)}{4N}, \quad 1 \leq k \leq N,
\end{equation} 
and persists for $|\gamma| > \gamma_{N,k}$ except possibly finitely many points on any 
compact interval of $\gamma$, where the pair coalesces into 
a double (semi-simple) pair of real eigenvalues. 
In particular, no complex eigenvalues exist for
$\gamma \in (-\tilde{\gamma}_N,\tilde{\gamma}_N)$, where
\begin{equation}
\tilde{\gamma}_N := \gamma_{N,k} = 2 \cos \frac{\pi (2N-1)}{4 N}.
\end{equation}
\label{theorem-linear-defect}
\end{theorem}

\begin{proof}
We set
$$
x_k = w_{2k-1}, \quad y_k = w_{2k}, \quad 1 \leq k \leq N
$$
and rewrite the linear eigenvalue problem (\ref{dnls-linear-closed}) in the equivalent form:
\begin{equation}
\label{dnls-linear-closed-equiv}
\left\{ \begin{array}{l}
2 \cos \theta x_k = y_{k-1} + y_{k} - i \gamma x_k, \\
2 \cos \theta y_k = x_{k} + x_{k+1} + i \gamma y_k, \end{array} \right. \quad 1 \leq k \leq N,
\end{equation}
subject to the modified boundary conditions $y_0 = x_1 e^{i \theta}$ and $x_{N+1} = y_N e^{i \theta}$. Expressing
$y_k$ from the second equation of the system (\ref{dnls-linear-closed-equiv}) by
$$
y_k = \frac{x_k + x_{k+1}}{2 \cos \theta - i \gamma}, \quad 1 \leq k \leq N-1, \quad
y_N = \frac{x_N}{e^{-i \theta} - i \gamma},
$$
and substituting these expressions to the first equation of the system, we obtain a second-order
difference equation
\begin{equation}
\label{second-order}
(\gamma^2 + 4 \cos^2 \theta) x_k = x_{k-1} + 2 x_k + x_{k+1}, \quad 1 \leq k \leq N,
\end{equation}
where the boundary conditions are now
\begin{equation}
\label{second-order-bc}
x_0 = (e^{i \theta} - i \gamma) e^{i \theta} x_1, \quad
x_{N+1} = \frac{e^{i \theta} x_N}{e^{-i \theta} - i \gamma}.
\end{equation}
Note that equations (\ref{second-order}) are obtained by multiplying every term by
$2 \cos \theta - i \gamma$, which is hence supposed to be non-zero.

The second-order difference equation (\ref{second-order}) admits an exact solution
\begin{equation}
\label{decomposition}
x_k = C_+ e^{i \alpha (k-1)} + C_- e^{-i \alpha (k-1)},
\end{equation}
where $\alpha$ is defined from the transcendental equation
\begin{equation}
\label{alpha}
\gamma^2 + 4 \cos^2 \theta = 2 + 2 \cos \alpha,
\end{equation}
and $(C_+,C_-)$ are non-zero solutions of the linear system following
from the boundary conditions (\ref{second-order-bc}):
\begin{eqnarray*}
e^{-i\alpha} C_+ + e^{i \alpha} C_- & = & (e^{i \theta} - i \gamma) e^{i \theta} (C_+ + C_-), \\
(e^{-i\theta} - i \gamma ) ( C_+ e^{i \alpha N} + C_- e^{-i \alpha N}) & = &
e^{i \theta} ( C_+ e^{i \alpha (N-1)} + C_- e^{-i \alpha (N-1)}).
\end{eqnarray*}
Note that for fixed $\gamma$, the values of $\theta$ are obtained from the characteristic equation
for this linear homogeneous system, after the values of $\alpha$ are excluded from the
transcendental equation (\ref{alpha}). Also note that only the values of $\theta$ with
${\rm Im}(\theta) > 0$ determine isolated eigenvalues of the linear stationary
PT-dNLS equation
(\ref{dnls-linear-defect}) by means of the representation (\ref{theta}).

After some algebraic manipulations, the characteristic equation for the linear system takes
the form
\begin{eqnarray*}
\left( \cos(\alpha N) \sin \theta \sin \alpha
- i \sin(\alpha N) \cos \theta (1- \cos \alpha) \right) ( \gamma  + 2 i \cos \theta ) = 0.
\end{eqnarray*}
The equation $\gamma  + 2 i \cos \theta = 0$ gives an artificial root corresponding to the value
$\alpha = \pi$, because the second-order difference equation (\ref{second-order}) is obtained by
multiplying every term by $2 \cos \theta - i \gamma$. Therefore, we drop this nonzero factor from
the characteristic equation and reduce it to the transcendental equation
\begin{eqnarray}
e^{2 i \theta} = \frac{\sin (N+1) \alpha - \sin N \alpha}{\sin N \alpha - \sin (N-1) \alpha} =
\frac{\cos \left(N + \frac{1}{2}\right) \alpha}{\cos \left(N - \frac{1}{2}\right) \alpha}.
\label{characteristic-eq}
\end{eqnarray}

{\bf Case $N = 1$:} Equation (\ref{characteristic-eq}) yields $e^{2 i \theta} = 2 \cos \alpha - 1$.
When this constraint is used in equation (\ref{alpha}), we obtain $e^{-2 i \theta} = 1 - \gamma^2$.
Since $\gamma \in \R$, we obtain the existence of a simple eigenvalue with ${\rm Re}(\theta) > 0$
for $\gamma^2 > 2$ and the bifurcation occurs for $\gamma = \sqrt{2}$ and corresponds to
$\theta = \pm \frac{\pi}{2}$, when $\cos \theta = 0$.\\

{\bf Case $N \geq 2$:} In a general case, we first consider bifurcations of complex values of
$\theta$ from real values of $\theta$. Hence, we set $\theta \in \mathbb{R}$ and realize from
(\ref{alpha}) that ${\rm Im}(\cos \alpha) = 0$, which implies either $\alpha \in \R$ or
$\alpha \in \pi k + i \R$ for any $k \in \mathbb{Z}$. In both cases, the characteristic equation
(\ref{characteristic-eq}) with real $\theta$ 
implies that either $\theta = 0$ or $\theta = \pm \frac{\pi}{2}$.

If $\theta = 0$, then the characteristic equation (\ref{characteristic-eq}) reduces to
the equation $\sin(\alpha N) = 0$, whereas equation (\ref{alpha}) implies that
$\gamma^2 = -2(1 - \cos \alpha) \leq 0$, which is outside of the parameter range
we are interested in. (Recall here that no eigenvalues with ${\rm Im} \theta > 0$
exist in the self-adjoint case with $\gamma = 0$.)

On the other hand, if $\theta = \pm \frac{\pi}{2}$,
then the characteristic equation (\ref{characteristic-eq}) reduces to
the equation $\cos(\alpha N) = 0$ (recall here that the artificial root $\alpha = \pi$ is neglected)
or equivalently,
$$
\alpha = \alpha_k := \frac{\pi (2k-1)}{2N}, \quad 1 \leq k \leq N.
$$
From equation (\ref{alpha}), we obtain that the bifurcation occurs at
$$
\gamma^2 = \gamma_k^2 := 4 \cos^2\left(\frac{\alpha_k}{2}\right) = 4 \cos^2\left( \frac{\pi (2k-1)}{2N}\right),
$$
which corresponds to the values (\ref{gamma-N-k}).

Next, we show that the values with ${\rm Im}(\theta) > 0$ correspond
to the range $\gamma^2 > \gamma_k^2$. This implies that a
new isolated eigenvalue $E$ of the linear stationary PT-dNLS equation (\ref{dnls-linear-defect})
with ${\rm Im}(E) \neq 0$ bifurcates from the value $E = 0$ that corresponds to
$\theta = \pm \frac{\pi}{2}$ at $\gamma = \pm \gamma_k$ and persists for all
$|\gamma| > \gamma_k$. A symmetric complex-conjugate
eigenvalue $\bar{E}$ exists by the PT symmetry of the stationary PT-dNLS
equation (\ref{dnls-linear-defect}). 

To show the above claim, we use the parametrization $z := e^{i \alpha}$,
$\zeta := e^{2i \theta}$, and $s := \gamma^2$. The system of transcendental
equations (\ref{alpha}) and (\ref{characteristic-eq}) becomes the system of
algebraic equations:
\begin{eqnarray*}
z + \frac{1}{z} - \zeta - \frac{1}{\zeta} & = & s, \\
\zeta z (1 + z^{2N-1}) - z^{2N+1} & = & 1.
\end{eqnarray*}
We know that for $s = s_0 := \gamma_k^2$, there exists a solution
of this algebraic system for $\zeta = -1$ and $z^{2N} = -1$. Therefore,
we consider the continuation of this complex-valued root in real values of $s$.
By computing the derivative in $s$, we obtain
\begin{eqnarray*}
\left[ \begin{array}{cc} \zeta^{-2} - 1 & 1 - z^{-2} \\
z (1 + z^{2N-1}) & \zeta (1 + 2 N z^{2N-1}) - (2N+1) z^{2N} \end{array} \right]
\left[ \begin{array}{c} \frac{d\zeta}{ds} \\ \frac{d z}{ds} \end{array} \right] =
\left[ \begin{array}{c} 1 \\ 0 \end{array} \right].
\end{eqnarray*}
For $s = s_0$, we obtain from this linear system that
$$
\frac{d \zeta}{ds} \biggr|_{s = s_0} = -\frac{2N z}{(1-z)^2} \biggr|_{s = s_0} = \frac{N}{2 \sin^2\left(\frac{\alpha_k}{2}\right)}
$$
and since $\zeta = e^{2 i \theta}$, this proves that $\frac{d}{ds} {\rm Im}(\theta) |_{\theta = \pm \frac{\pi}{2}} > 0$.
By continuity of roots of the algebraic system above, ${\rm Im}(\theta)$ remains positive for $s > s_0$,
that is, for $\gamma^2 > \gamma_k^2$, near the bifurcation point. 

We have shown that roots with ${\rm Im}(\theta) > 0$ bifurcate from the 
points $\theta = \pm \frac{\pi}{2}$ and remain in the upper half-plane 
for all $\gamma^2 > \gamma_k^2$. These roots correspond to complex 
eigenvalues $E$ if ${\rm Re}(\theta) \neq 0$ (in which case $E > 2$) 
or if ${\rm Re}(\theta) \neq \pm \pi$ (in which case $E < 2$). Both situations 
can not be a priori excluded, however, we have here two facts:
\begin{itemize}
\item When ${\rm Re}(\theta) = 0$ or ${\rm Re}(\theta) = \pm \pi$, 
a pair of complex-conjugate eigenvalues $E$ coalesce at the real line 
into a double (semi-simple) eigenvalue, because of the PT-symmetry 
with ${\rm Im}(\theta) \neq 0$ generates two linearly independent eigenvectors 
for the same real eigenvalue. 

\item Roots $\theta$ are analytic with respect to parameter $\gamma$ 
by analytic dependence of roots of the algebraic system. 
\end{itemize}

Combining these two facts together, we realize that that the double (semi-simple) 
eigenvalues $E$ can not split along the real axis (as each real eigenvalue 
after splitting would then become a double eigenvalue by the PT-symmetry). 
And they can not persist on the real axis as a double eigenvalues because of 
analyticity of the parameter continuation of roots $\theta$ with respect to $\gamma$. 
Therefore, these double real roots split back to the complex domain. In addition, 
analyticity of the parameter continuations guarantees that there are finitely many 
points on any compact interval in $\gamma$, where the pairs of 
complex-conjugate eigenvalues can coalesce at the real axis.
The proof of the theorem is complete.
\end{proof}

We list some numerical values of the phase transition thresholds:
\begin{eqnarray*}
\tilde{\gamma}_1 & = & 2 \cos \frac{\pi}{4} = \sqrt{2}, \\
\tilde{\gamma}_2 & = & 2 \cos \frac{3 \pi}{8} \approx 0.765, \\
\tilde{\gamma}_3 & = & 2 \cos \frac{5 \pi}{12} \approx 0.518,
\end{eqnarray*}
with $\lim_{N \to \infty} \tilde{\gamma}_N = 0$. Note that
$\tilde{\gamma}_N > \gamma_N$ for any $N \in \N$.

\subsection{Stationary states: bifurcations from the anti-continuum limit}

We now consider the stationary states, which satisfy the nonlinear stationary
PT-dNLS equation
\begin{eqnarray}
E w_n = w_{n+1} + w_{n-1} + i \gamma (-1)^n \chi_{n \in S_N} w_n + |w_n|^2 w_n, \quad n \in \Z.
\label{eqn1}
\end{eqnarray}
We explore the large-amplitude limit similarly to Section 3.3. Hence,
we set $E = \frac{1}{\delta}$ and ${\bf w} = \frac{\bf W}{\sqrt{\delta}}$,
where $\delta$ is a small positive number. The stationary PT-dNLS equation (\ref{eqn1})
is rewritten in the equivalent form:
\begin{eqnarray}
(1-|W_n|^2) W_n= \delta \left( W_{n+1} + W_{n-1} + i \gamma
(-1)^n \chi_{n \in S_N} W_n \right), \quad n \in \Z.
\label{eqn2}
\end{eqnarray}
We consider the PT-symmetric solutions with ${\bf W} = P \bar{\bf W}$, where $P$
is given by $[P {\bf W}]_n = W_{2N+1-n}$. Note that the choice of $n_0$ in $P$
given by Corollary \ref{corollary-PT-symmetry} is adjusted to the center of the PT-symmetric defect.
Therefore, the existence of the PT-symmetric solutions can be considered in the framework of
the following two subsystems:
\begin{eqnarray}
(1-|W_n|^2) W_n = \delta \left( W_{n+1} + W_{n-1} + i \gamma (-1)^n W_n \right), \quad 1 \leq n \leq N,
\label{eqn2a}
\end{eqnarray}
and
\begin{eqnarray}
(1-|W_n|^2) W_n = \delta (W_{n+1} + W_{n-1}), \quad n \leq 0,
\label{eqn2b}
\end{eqnarray}
subject to the boundary condition $W_{N+1} = \bar{W}_N$. The following result
is similar to the result of Theorem \ref{theorem-nonlocal-bifurcation}.

\begin{theorem}
\label{theorem-soliton}
For any $\gamma \in (-1,1)$,
the nonlinear stationary PT-dNLS equation (\ref{eqn2}) in the limit of small positive $\delta$
admits $2^N$ PT-symmetric solutions ${\bf W} = P \bar{\bf W} \in l^2(\Z)$ (unique up to a gauge transformation) such that,
for sufficiently small $\delta$, the map $\delta \to {\bf W}$ is $C^{\infty}$ at each solution and
there is a positive $\delta$-independent constant $C$ such that
\begin{equation}
\label{bound-soliton}
\left| {\bf W} - {\bf W}_0 \right| \leq C \delta,
\end{equation}
where ${\bf W}_0$ is a solution of Theorem \ref{theorem-nonlocal-bifurcation} (after rescaling).
\end{theorem}

\begin{proof}
We consider small solutions of the subsystem (\ref{eqn2b}) for a given $W_1 \in \C$ and small $\delta \in \R$.
Parameter $\gamma \in \R$ is fixed.
The nonlinear system represents a bounded $C^{\infty}$ map from
$({\bf W}_-,W_1,\delta) \in l^2(\mathbb{Z}_-) \times \C \times \R$ to $l^2(\mathbb{Z}_-)$,
where $\Z_-$ is the set of negative integers including zero. For $\delta = 0$ and arbitrary
$W_1 \in \C$, ${\bf W}_- = {\bf 0}$ is a root of the nonlinear map and
the Jacobian with respect to ${\bf W}_-$ at ${\bf W}_- = {\bf 0}$ is invertible.
By the Implicit Function Theorem, for all $W_1 \in \C$ and small $\delta \in \R$, there is a unique solution
of the nonlinear system (\ref{eqn2b}) such that the map $(W_1,\delta) \to {\bf W}_-$ is $C^{\infty}$
and there is a positive $\delta$-independent constant $C$ such that $\| {\bf W}_- \|_{l^2(\mathbb{Z}_-)} \leq C |\delta|$.

Substituting $W_0$ from the map constructed above
to the first equation of the subsystem (\ref{eqn2a}), we close the system
at $N$ nonlinear equations for $\{ W_n \}_{1 \leq n \leq N}$. The only difference from
the $N$ nonlinear equations considered in the proof of Theorem \ref{theorem-nonlocal-bifurcation}
is the boundary condition for given $W_0$, however, $W_0$ is small as $W_0 = \mathcal{O}(\delta)$.
The two applications of the Implicit Function Theorem developed 
in the proof of Theorem \ref{theorem-nonlocal-bifurcation}
apply directly to our case and yield the assertion of this theorem. The bound
(\ref{bound-soliton}) follows from bounds (\ref{bound-1}) and (\ref{bound-2}).
\end{proof}

\begin{remark}
A remark similar to Remark \ref{remark-other-branches} applies on the infinite
lattice as well. Besides localized states of Theorem \ref{theorem-soliton},
for any $N \geq 2$ and $1 \leq M \leq N$, there exist additional soliton states
such that $|W_n|^2 \approx 1$ as $\delta \to 0$ for $N-M+1 \leq n \leq N+M$
and $|W_n|^2 \approx 0$ as $\delta \to 0$ for $n \leq N-M$ and $n \geq N+M+1$.
These stationary states are supported at $2M$ sites near the central sites in 
the set $S_N$.
\label{remark-other-solitons}
\end{remark}

We give details of the perturbative expansions
for the two (most fundamental) discrete solitons supported by the dimer defect for $N = 1$.
The subsystems (\ref{eqn2a}) and (\ref{eqn2b}) are rewritten explicitly as follows:
\begin{eqnarray}
(1-|W_1|^2) W_1 &=& \delta \left( W_0 + \bar{W}_1 - i \gamma W_1 \right),
\label{eqn3} \\
(1-|W_n|^2) W_n &=& \delta \left( W_{n+1} + W_{n-1} \right), \quad \quad n \leq 0.
\label{eqn4}
\end{eqnarray}
The perturbation expansion
\begin{eqnarray}
\label{perturbation}
W_n= W_n^{(0)} + \delta W_n^{(1)} + \mathcal{O}(\delta^2)
\end{eqnarray}
allows us to compute at the leading order $W_n^{(0)} = e^{-i \theta} \delta_{n,1}$,
where $\theta \in [0,\pi]$ is arbitrary at this point. At the $\mathcal{O}(\delta)$ order,
we obtain the equations:
\begin{eqnarray}
W_n^{(1)} &=& W_{n+1}^{(0)} + W_{n-1}^{(0)}, \quad \quad n \leq 0
\label{eqn5} \\
-\left(e^{-2 i \theta} \bar{W}_1^{(1)} + W_1^{(1)} \right)
&=& e^{i \theta} - i \gamma e^{-i \theta}, \quad \quad n =1
\label{eqn6}
\end{eqnarray}
Set $W_1^{(1)} := Z e^{-i \theta}$ and rewrite (\ref{eqn6}) as
$-(\bar{Z} + Z) = e^{2 i \theta} - i \gamma$. The solvability condition is
$\sin(2 \theta)= \gamma$ and it gives exactly two values for $\theta \in [0,\pi]$
for any $\gamma \in (-1,1)$. Then, the first-order correction term is found explicitly as
follows:
\begin{eqnarray*}
W_n^{(1)} = -\frac{1}{2} e^{-i \theta} \cos(2 \theta) \delta_{n,1} + e^{-i \theta} \delta_{n,2}.
\end{eqnarray*}
The perturbation expansion (\ref{perturbation}) can be continued to higher orders of $\delta$ 
thanks to $C^{\infty}$ smoothness in Theorem \ref{theorem-soliton}.
Hence, we obtain two branches of soliton states supported by the dimer defect.

\begin{remark}
The phase transition threshold for $S_1$ is $\tilde{\gamma}_1 = \sqrt{2}$,
whereas the soliton states of Theorem \ref{theorem-soliton} are only constructed for $\gamma \in (-1,1)$
in the limit $E \to \infty$. Numerical studies (see Figure \ref{fig5})
show that the solution states exist for $\gamma \in (-1,1)$
for any fixed value of $E$. \label{remark-observation}
\end{remark}

\subsection{Numerical results}

We now test these analytical results for the PT-symmetric chain with embedded defects.

First, we consider the linear limit of
such chains and examine the corresponding PT-symmetry phase
transitions expected to occur at $\tilde{\gamma}_N$. In particular,
in Fig.~\ref{lin_fig1}, we consider
the case of $N=1$, i.e., a single embedded dimer which has been
predicted to have a PT-phase transition at $\tilde{\gamma}_1=\sqrt{2}$,
for the infinite lattice. In the figure, we can clearly discern
the relevant transition (the corresponding vertical dashed
line shows the theoretical prediction of $\tilde{\gamma}_1$).
Nevertheless, for the finite lattice considered (here $800$ sites are used),
an additional bifurcation is observed at $\gamma=1$ (see also
the works of~\cite{Ambroise,Pelin2}). This bifurcation
is suppressed at the infinite lattice limit, but for a finite chain a
``bubble'' of complex eigenvalues arises around $E=0$ (which is the middle
point of the spectral band). This bubble keeps expanding and slowly increasing
in imaginary part between $\gamma=1$ and $\gamma=\sqrt{2}$ (notice that
this growth is barely visible in the linear scale of the figure but
it is noticeable in the logarithmic scale of the inset). At
the latter critical point, the rapid growth of the
isolated unstable eigenvalue pair becomes dominant for
the instability of the lattice with an embedded dimer.

\begin{figure}[tbp]
\includegraphics[width=78mm,keepaspectratio]{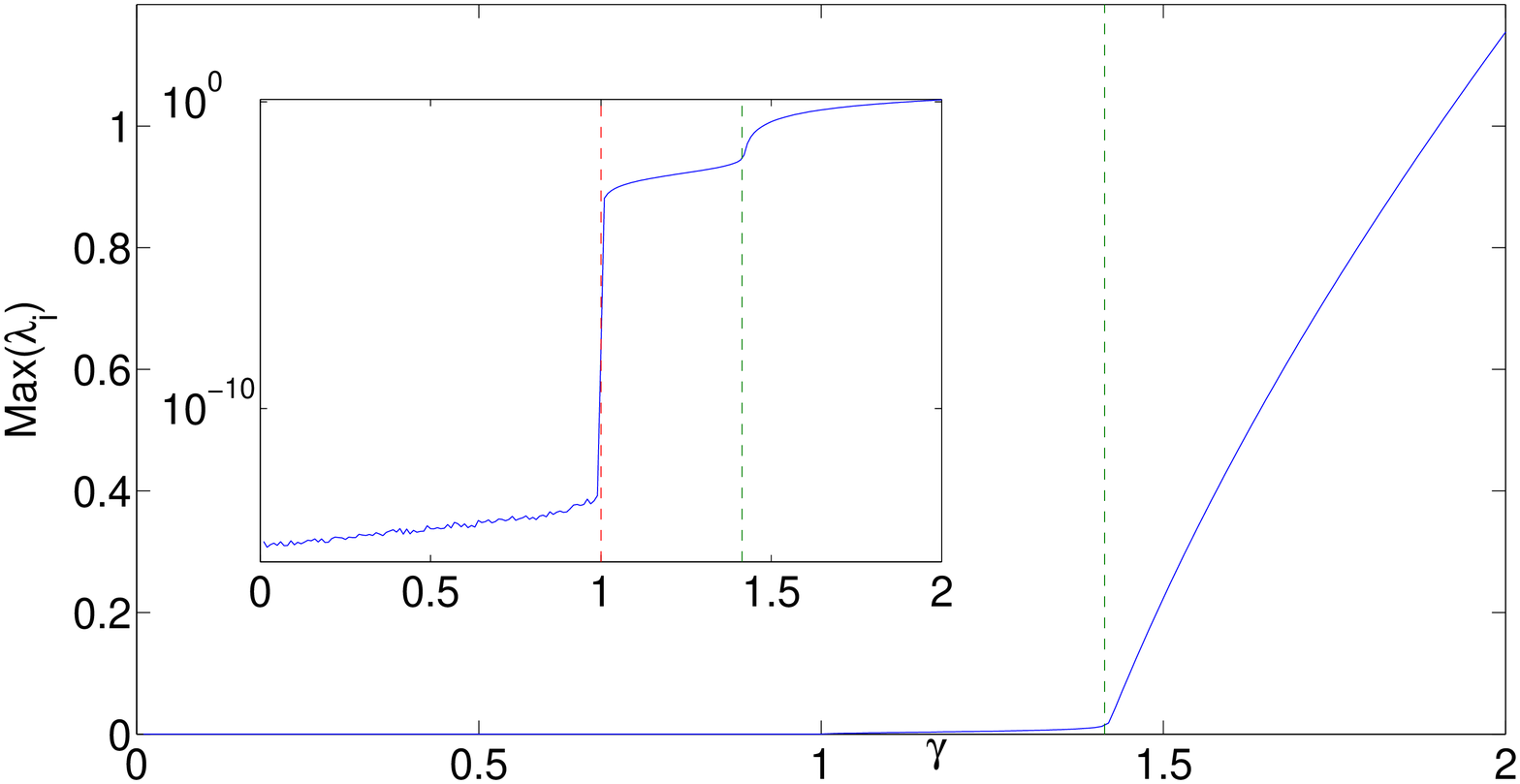}
\includegraphics[width=78mm,keepaspectratio]{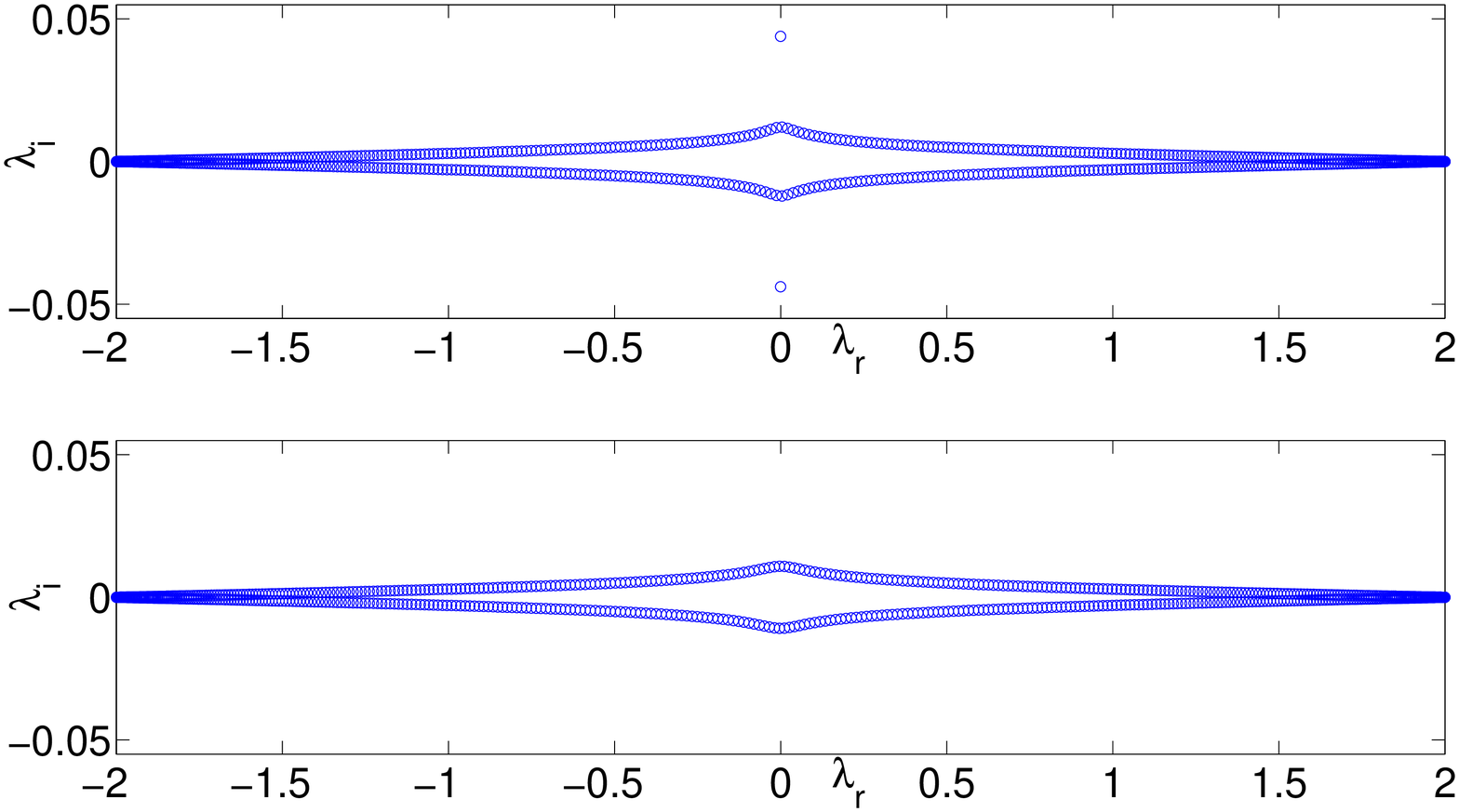}
\caption{The left panel shows the dependence of the maximal imaginary
part $\lambda_i$ of the eigenvalues $\lambda=\lambda_r + i \lambda_i$
versus $\gamma$ for $N = 1$. While this is shown in a linear
scale with the green dashed line denoting the
bifurcation point $\tilde{\gamma}_1=\sqrt{2}$, the inset
in a semi-logarithmic scale demonstrates the finite-size instability
emerging at $\gamma=1$, which is denoted by the red dashed line.
The right panel shows the spectrum of the linear PT-dNLS equation
(\ref{dnls-linear-defect}) for $\gamma=1.43$ (top)
and $\gamma=1.4$ (bottom).}
\label{lin_fig1}
\end{figure}

Similar conclusions can be drawn for the case of embedded quadrimer
($N=2$) and hexamer ($N=3$) from Fig.~\ref{lin_fig2}.
The figures reveal, however, that in this case in addition to the
actual (infinite chain) critical points of $\tilde{\gamma}_2 \approx 0.765$
and $\tilde{\gamma}_3 \approx 0.518$, respectively, there are multiple
additional points of weak instability emergence due to finite
size effects. Such features are noticeable due to bifurcations
of instability bubbles at the edges of the spectral band
(at $E=2$ and $E=-2$), at $\gamma \approx 0.46$ for $N=2$
and at $\gamma \approx 0.30$ for $N=3$. Additional bubbles
emerge near the middle point of the spectral band (at $E = 0$),
at $\gamma \approx 0.61$ for $N=2$ and at $\gamma \approx 0.45$
for $N=3$. 

\begin{figure}[tbp]
\includegraphics[width=78mm,keepaspectratio]{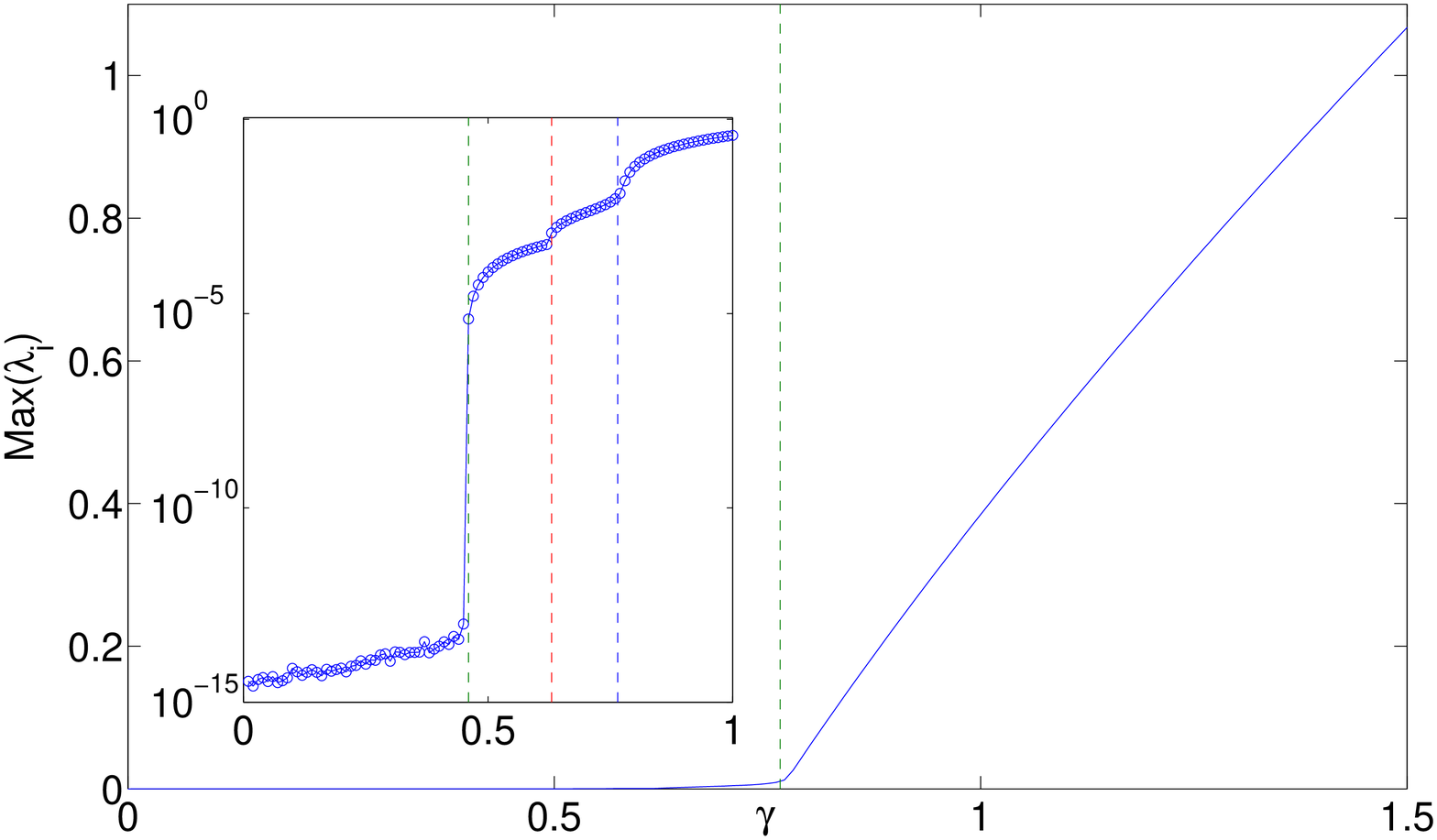}
\includegraphics[width=78mm,keepaspectratio]{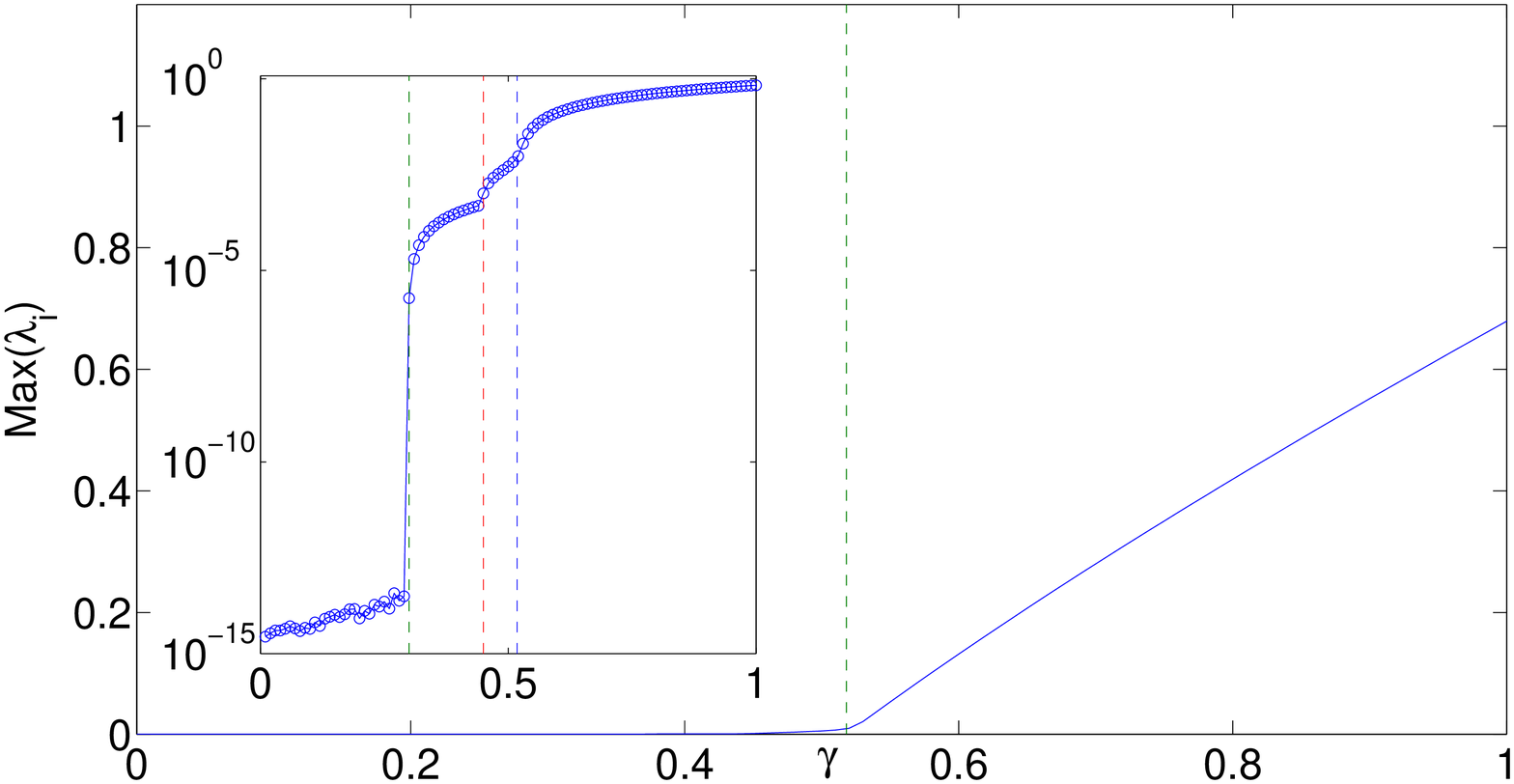}
\caption{The case of $N=2$ (left panel) and $N=3$ (right panel) for a finite lattice
of $800$ sites. The panels are similar to the left panel of Fig.~\ref{lin_fig1},
but feature two weak instabilities due to finite size effects (green
and red dashed line), in addition to the strong instability at $\tilde{\gamma}_N$ 
(blue dashed line).}
\label{lin_fig2}
\end{figure}

Next, we turn to  the existence of nonlinear stationary solutions in
the PT-symmetric dimer ($N=1$) embedded in the nonlinear chain, as shown in
Fig.~\ref{fig5}. We have identified two branches of
localized states starting from the
Hamiltonian limit where $2 \theta=0$ and $2 \theta=\pi$ (i.e.,
in-phase and anti-phase discrete solitons, respectively). Indeed,
the left panel shows the phase difference between $w_1$ and $w_2$
for the PT-symmetric embedded dimer.
These results clearly illustrate that at the Hamiltonian limit
of $\gamma=0$, the system starts from the two well-known in phase and
out-of-phase solutions~\cite{pkf}. The former (lower norm one)
is unstable for our focusing nonlinearity, while the latter one is spectrally
stable. Interestingly, in accordance with what is known for an
isolated dimer~\cite{Li,Ramezani,Sukh}, these two solutions
disappear in a saddle-center bifurcation at $\gamma=1$.
In fact, both this increased proximity and the eventual collision
and disappearance are captured very accurately by the
solvability condition $\sin(2 \theta)=\gamma$.
The resulting angle from the numerical computation and from the essentially
coincident analytical prediction are shown in the left panel
of Fig.~\ref{fig5}. Note that while the results are shown in Fig.~\ref{fig5} for $\delta = 0.01$,
they remain similar not only for smaller values of $\delta$ (such as e.g. $0.001$),
but even for larger values up to $\delta=1$ that were considered.
In particular,  the point of the saddle-center bifurcation $\gamma = 1$ has been found to be identical
for other values of $\delta$. This confirms the observation
made in Remark \ref{remark-observation}.

In the Hamiltonian case of $\gamma=0$, our linear
stability results for these branches fall back on the analysis
of~\cite{pkf}.  In fact, we retrieve the exact same condition,
for the leading order eigenvalue correction, namely
$\lambda^2/\delta=4 \cos(2 \theta)$.
This expression reveals the instability of the in-phase mode
and the stability (for our focusing nonlinearity) of the out-of-phase
one. In the presence of gain/loss (i.e., for finite $\gamma$), the fundamental
difference lies in the existence condition which mandates that
$\sin(2 \theta)=\gamma$ and therefore the eigenvalues of the
in-phase unstable state approach the origin (from the real axis),
as do the ones of the out-of-phase marginally stable state (from the
side of the imaginary axis) as $\gamma$ is increased towards $1$. 
These eigenvalue pairs end up colliding
at the origin at the point of the PT-phase transition at $\gamma=1$.
The trajectories of the two sets of eigenvalues are shown in the right
panel of Fig.~\ref{fig5}. It can be seen that similarly to the Hamiltonian
case of~\cite{pkf}, the agreement is better for the unstable branch of
real eigenvalues. Nevertheless, for both branches the comparison of
computation and analysis is highly favorable.

\begin{figure}[tbp]
\includegraphics[width=78mm,keepaspectratio]{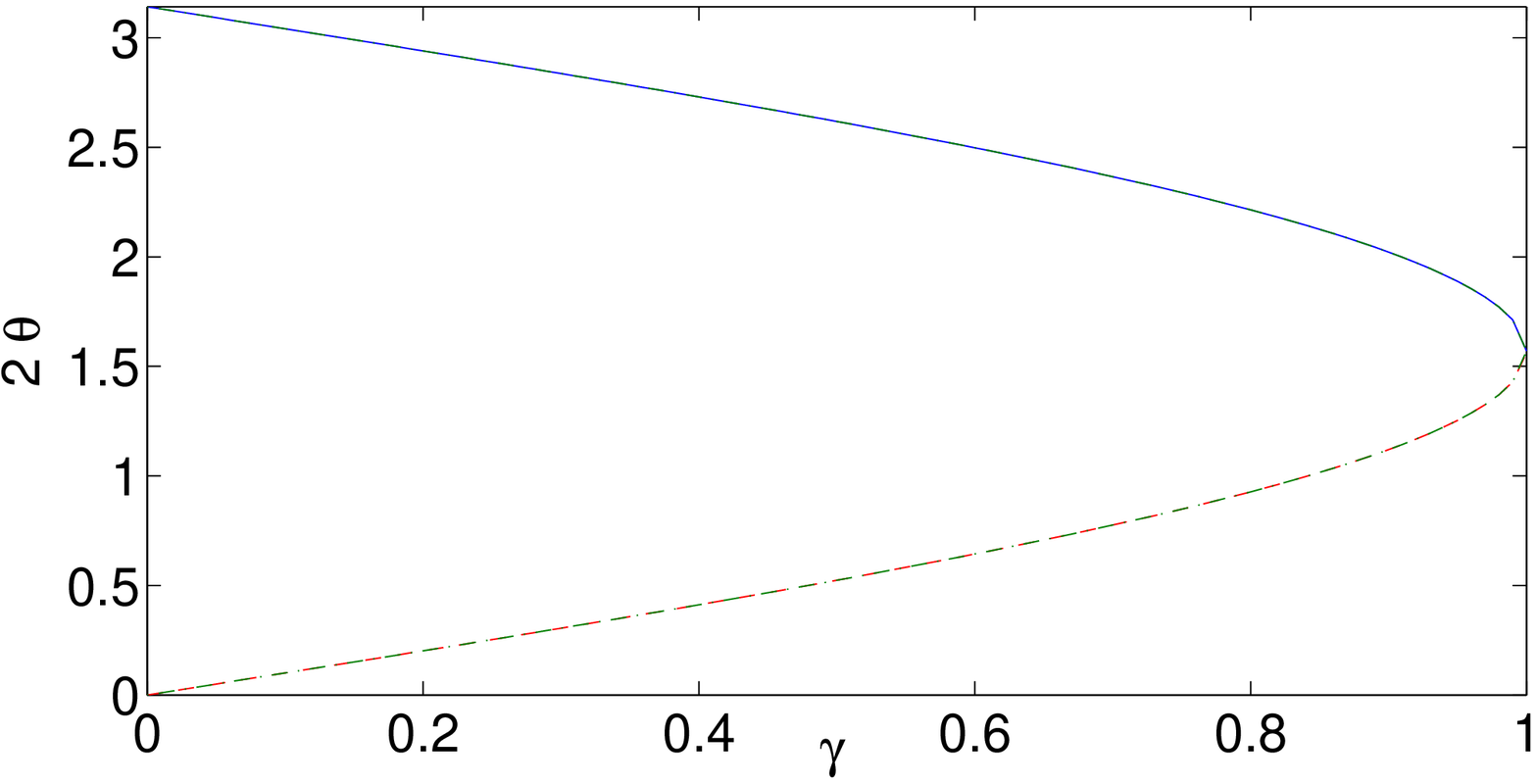}
\includegraphics[width=78mm,keepaspectratio]{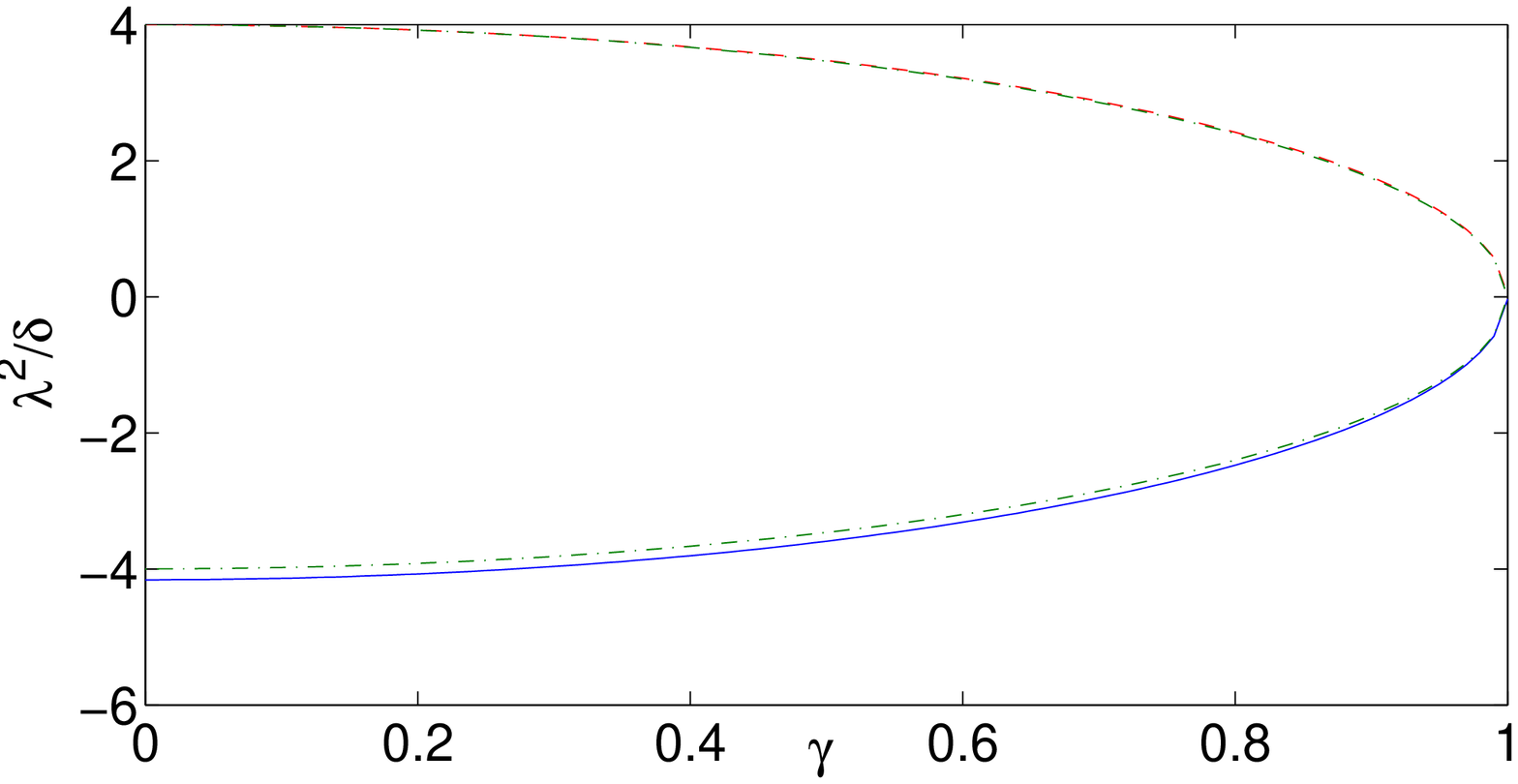}
\caption{The left panel shows the relative
phase $2 \theta$ between the two central sites obtained by the
numerical computation and the solvability condition
$\sin(2 \theta)=\gamma$ (a green dash-dotted line).
The right panel shows the squared eigenvalue of the linearized PT-dNLS
equation at the discrete soliton versus $\gamma$ for each of the two branches.
The stable (blue solid) branch of the anti-phase solution (at $\gamma=0$) has
a negative $\lambda^2$, while the unstable (red dashed) branch of
the in-phase solution (at $\gamma=0$) has a positive $\lambda^2$.
The corresponding theoretical predictions are shown by green dash-dotted
lines in very good agreement with the numerical results. }
\label{fig5}
\end{figure}

For the case of $N=1$, we now turn
to numerical simulations, in order to briefly discuss
the difference between the manifestation of the instability of the in-phase 
localized states for different values of $\gamma$. 
Two prototypical examples are shown in Fig.~\ref{fig3}.
Both panels illustrate the evolution of the two most central, maximum
amplitude sites of the solution. In the Hamiltonian case of $\gamma=0$
shown in the left panel, it can be seen that small perturbations give rise
to an amplified symmetry breaking in the dynamics. Nevertheless, the
conservative nature of the ensuing dynamics ascertains the rapid
saturation of this symmetry breaking and the eventual oscillations
that arise lead to a (nearly)
periodic alternation between a symmetric and a symmetry broken state.
On the other hand, for $\gamma = 0.5$ we observe a drastically
different manifestation of this instability shown in the right panel. More specifically, the site
associated with gain grows indefinitely in a nearly exponential form,
as illustrated in the inset. At the same time, the site associated with 
damping decreases in amplitude in a similar fashion. 

\begin{figure}[tbp]
\includegraphics[width=78mm,keepaspectratio]{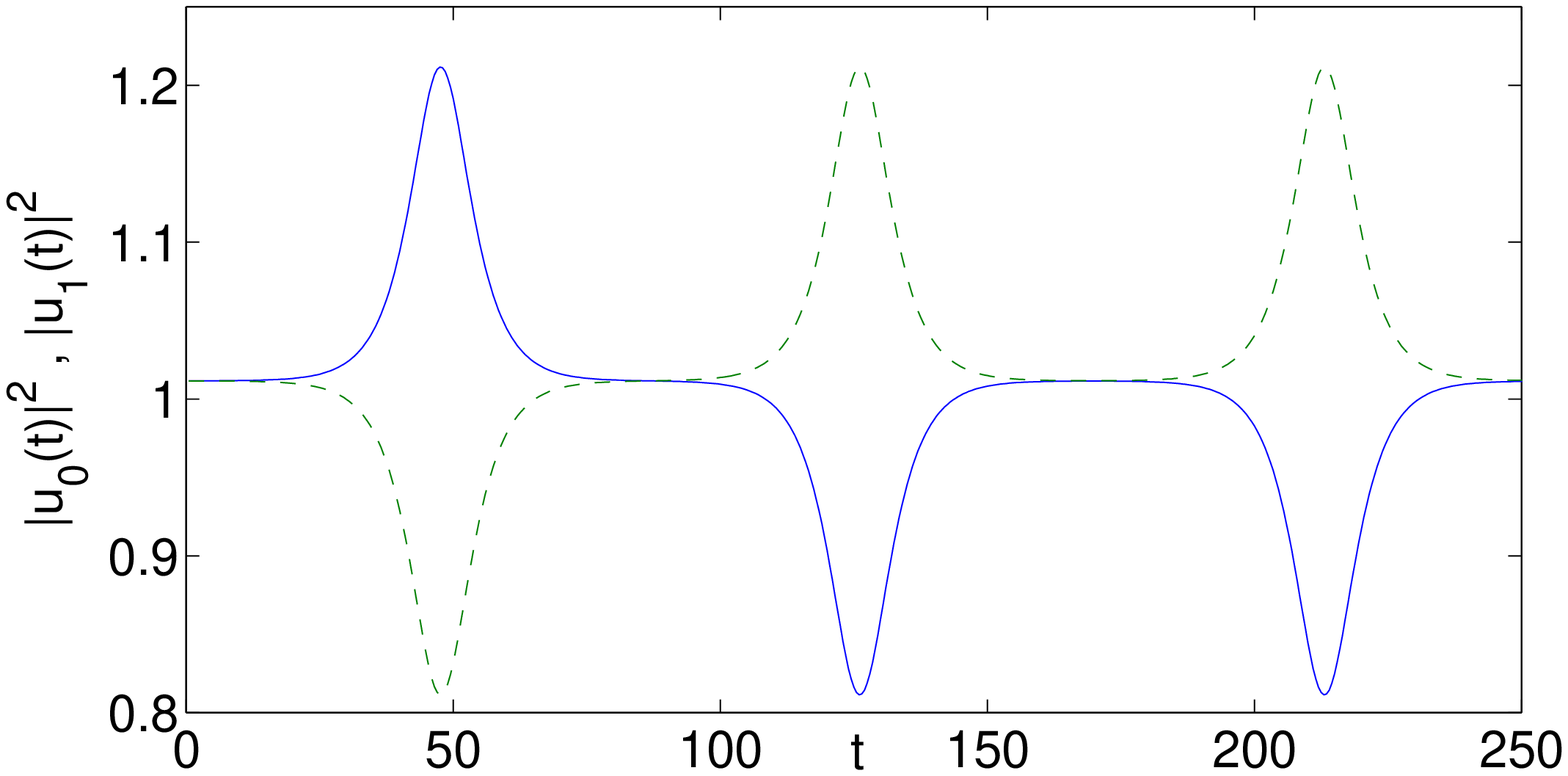}
\includegraphics[width=78mm,keepaspectratio]{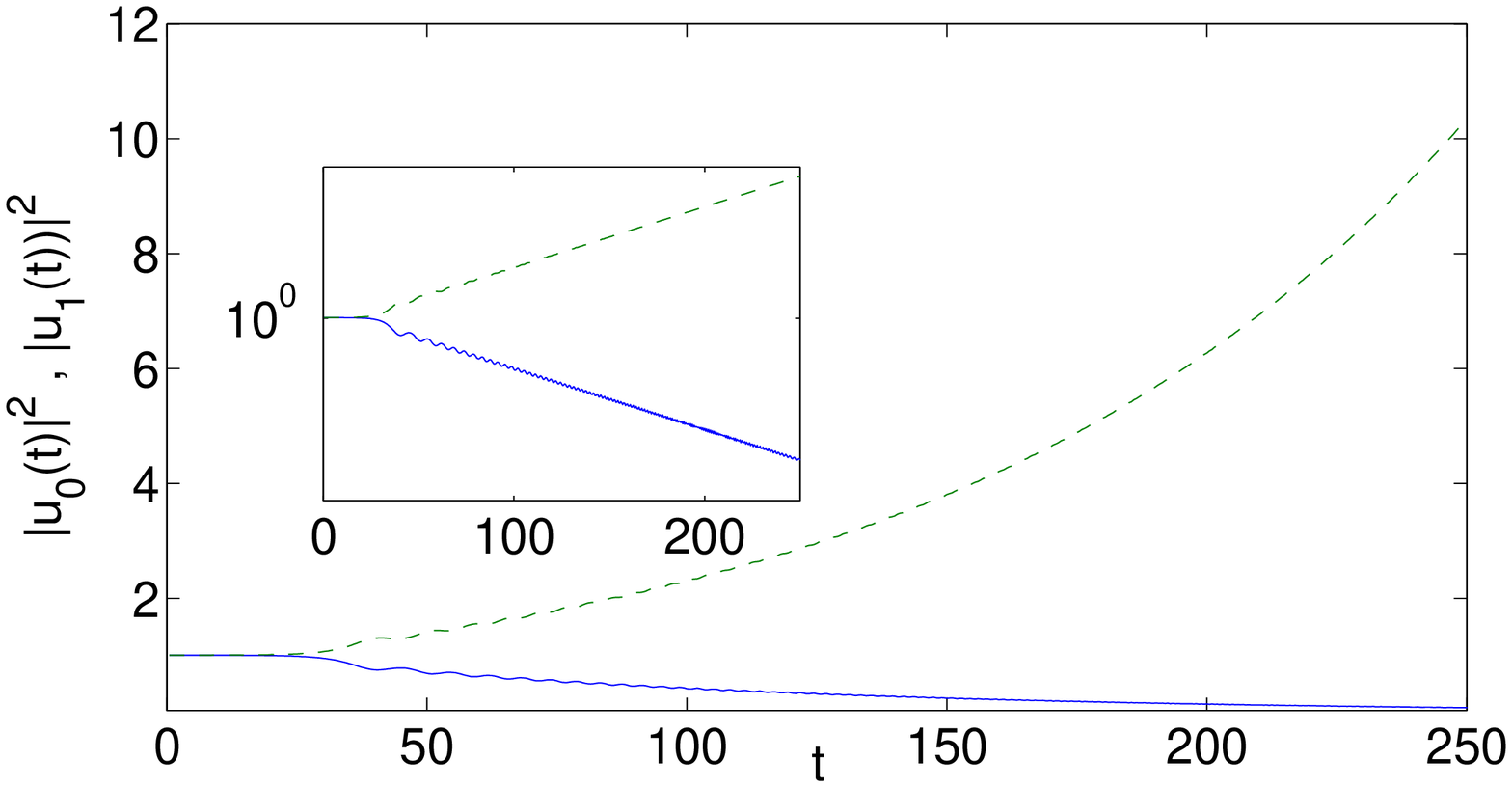}
\caption{The left panel shows the evolution of the unstable in-phase
solution for $\gamma=0$. The amplitude of the central-most two nodes
is shown as a function of time.
The symmetry-breaking results from the amplification
of a noise in the initial data and  manifests the instability via a
(nearly) periodic recurrence between symmetry-restored and symmetry-broken
phases. The right panel shows the same instability for $\gamma=0.5$.
The instability leads to an amplification
of the gained site amplitude and a decay of the damped site
amplitude.}
\label{fig3}
\end{figure}

\begin{figure}[tbp]
\includegraphics[width=78mm,keepaspectratio]{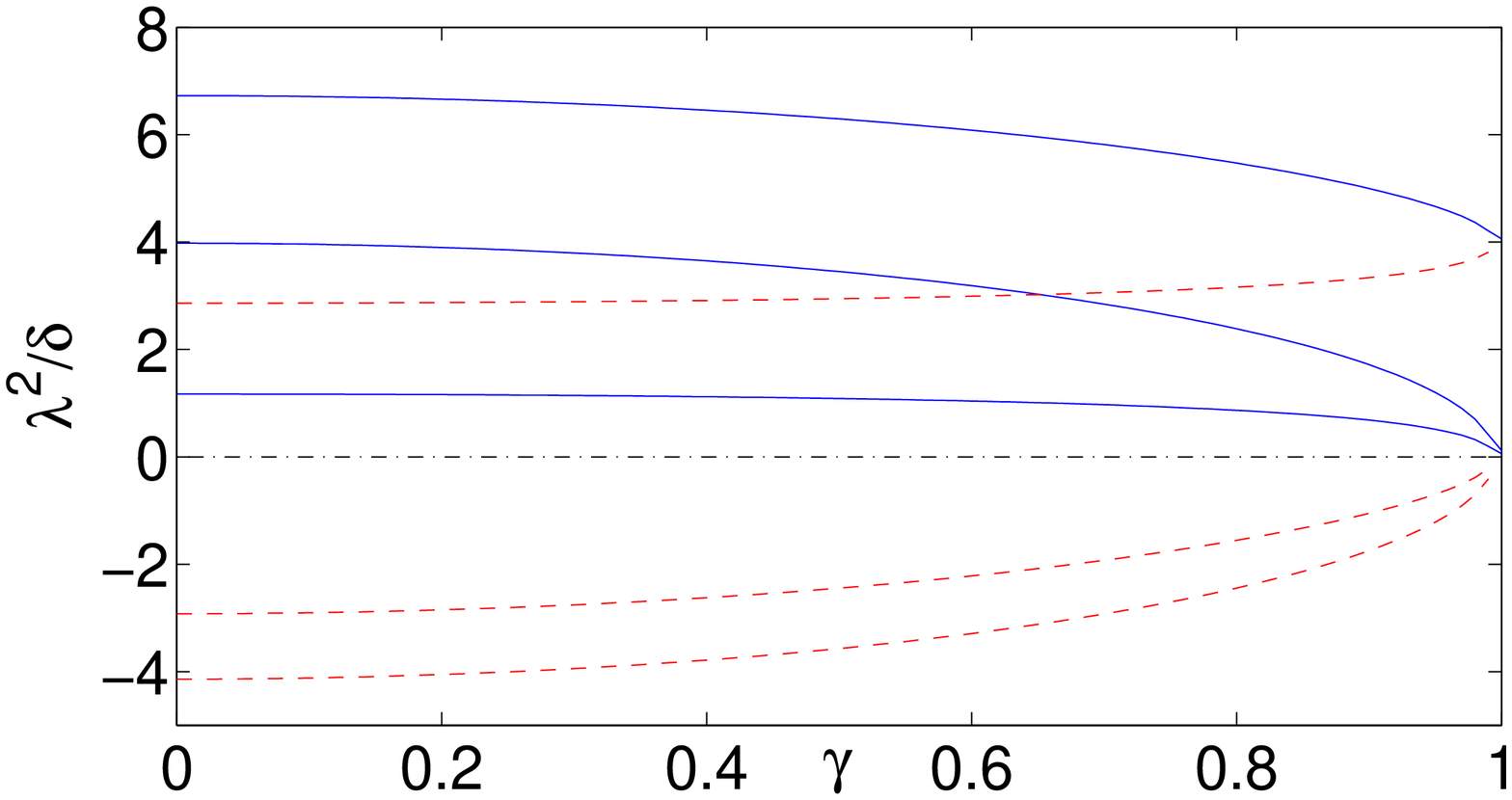}
\includegraphics[width=78mm,keepaspectratio]{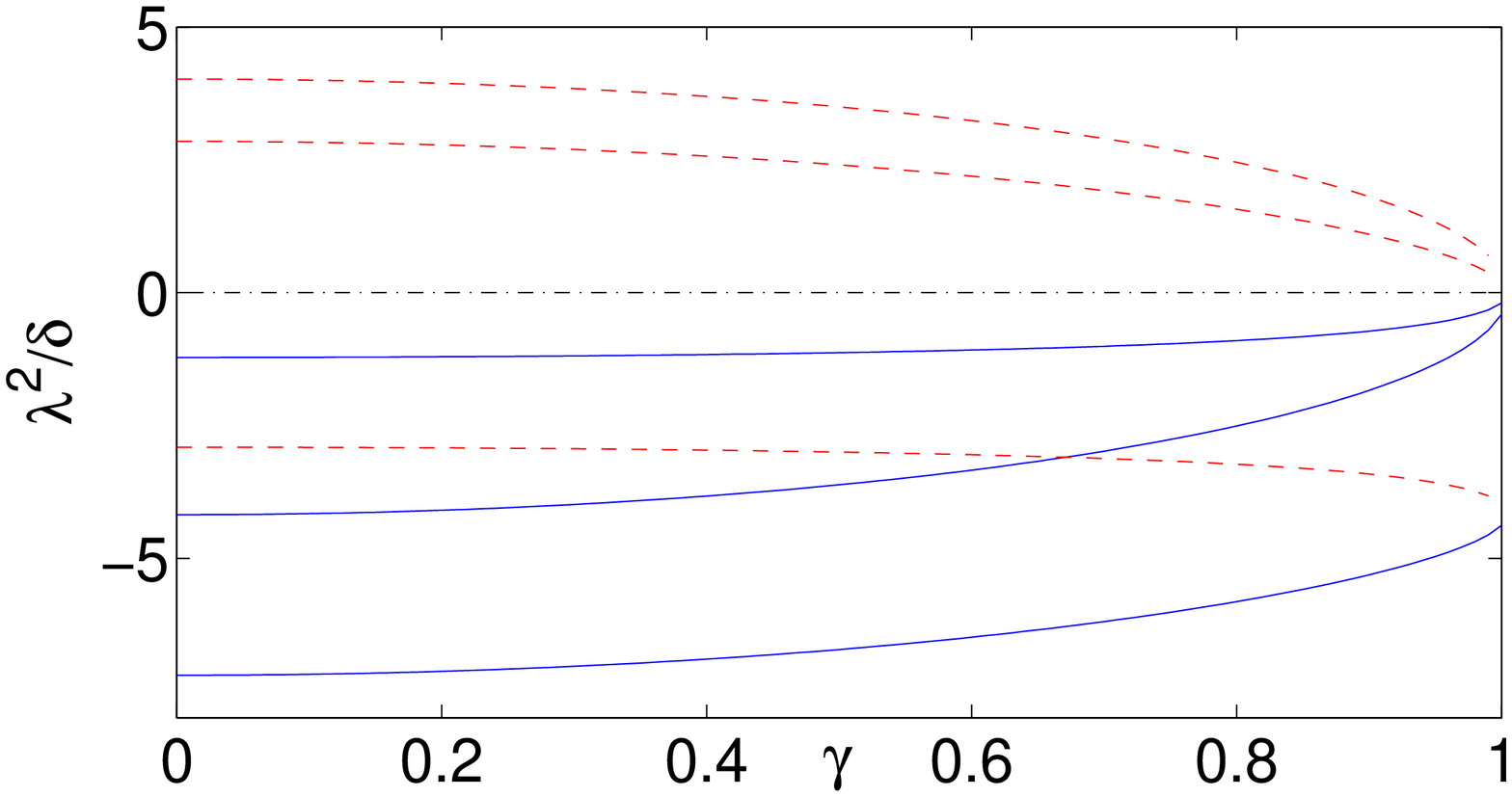}
\caption{Stability analysis results for the case of $N=2$. The left
panel depicts the $++++$ solution (which has three real eigenvalue
pairs) by solid (blue) lines, as well as the $+--+$ solution (which has two imaginary
and one real pair) by dashed (red) lines.
The right panel depicts the $+-+-$ solution
(which has three imaginary eigenvalue pairs) by solid (blue) lines, as well
as the $++--$ solution (which has two real pairs and one imaginary)
by dashed (red) lines. Each pair of these configurations collides
and disappears at $\gamma=1$.}
\label{fig4}
\end{figure}

Lastly, we address the bifurcation and stability results for the case
of $N=2$ i.e., for an embedded quadrimer. In this case, the eigenvalues
of the linearization problem are shown in Fig.~\ref{fig4}. We have examined four principal
configurations with all four sites excited (there exist also
configurations with two sites excited as for $N=1$, in accordance
with Remark \ref{remark-other-solitons}). In agreement with Theorem \ref{theorem-nonlocal-bifurcation},
these four configurations are coded as $++++$ (when all sites excited in
phase), $+-+-$ (when all sites are out of phase with their
immediate neighbors),
$++--$, and $+--+$. As can be seen in the figure and as is known
from the Hamiltonian limit of $\gamma=0$, the only spectrally stable
among these configurations is the out of phase configuration $+-+-$  with three imaginary
eigenvalue pairs, while the in-phase configuration $++++$ is the most unstable
with three real pairs. Upon continuation over $\gamma$, for both of
these configurations, two of the eigenvalue pairs move towards zero (which
they reach as $\gamma \rightarrow 1$), while one remains real for the in-phase, and imaginary for the out
of phase. Remarkably, pairwise these configurations collide and disappear in the limit of $\gamma \rightarrow 1$.
More specifically, the in-phase configuration $++++$ collides with the configuration
$+--+$. Similarly the out-of-phase configuration $+-+-$
collides with the configuration $++--$.
It should also be noted that we weren't able to continue any
asymmetric mixed phase configurations (such as e.g. $+++-$,
$+---$, $-+--$, $--+-$ or their opposite parity variants)
past the Hamiltonian limit of $\gamma=0$.

\section{Conclusion}

In this paper, we examined two distinct scenaria for PT-symmetric
dynamical lattices. In the first one, $N$-site PT-symmetric chains
were considered as a finite-dimensional dynamical system.
In the second one, we have considered the embedding of 
the finite PT-symmetric system as a defect in
an infinite dNLS lattice. In both cases, we have examined 
the linear problem, explicitly computing the corresponding eigenvalues
and identifying the strengths of $\gamma$ beyond which instabilities
(and the phase transition breaking the PT-symmetry) 
arise due to real eigenvalues. We have also considered
the nonlinear states when they emerge from the linear
limit, as well as when they arise from a  highly nonlinear
limit under suitable rescaling
(analogous to the anti-continuum limit of the standard dNLS
lattice). In that case, we argued about the disparity of the branch
counts in these two limits (for general $N$),
which suggests the existence of a number of bifurcations, such as saddle-center ones, at intermediate
values of the corresponding parameter. 

In the case of the infinite PT-symmetric PT-dNLS equation,
\begin{eqnarray}
i \frac{d u_n}{dt} = u_{n+1} - 2 u_n + u_{n-1}
+ i \gamma (-1)^n u_n + |u_n|^2 u_n, \quad n \in \Z,
\label{dnls-infinite}
\end{eqnarray}
we note that the phase transition threshold is now set at $\gamma_{N \to \infty} = 0$.
Therefore, for any $\gamma \neq 0$, the linear PT-dNLS equation is unstable
with a complex-valued continuous spectrum. Nevertheless, we can still obtain existence of
stationary localized states (discrete solitons) in the large-amplitude limit
for $\gamma \in (-1,1)$ for any of the configurations described
in Theorems \ref{theorem-nonlocal-bifurcation} and \ref{theorem-soliton}. Moreover,
the discrete soliton can be centered at any site $n_0 \in \Z$ because of the
shift invariance of the PT-dNLS equation (\ref{dnls-infinite}).

There are other numerous directions in which one can envision generalizations
of the present study. In the present work, we considered the case
where there is a single parameter $\gamma$, for each of
the $N$ pairs of sites with gain and loss. However, it is also 
relevant to generalize such considerations to
the case of many independent parameters for such sites~\cite{ZK}. On the other hand,
one can consider generalizations of the present setting that aim towards
the case of higher dimensionality. Arguably, the simplest such generalization
concerns the setting of two one-dimensional coupled (across each of their
sites) chains in the form of a railroad track as in~\cite{suchkov} and
the consideration of multi-site excitations therein. However, the
genuinely higher dimensional problem and the examination of
generalizations of plaquette configurations~\cite{Guenter}, whereby
the potential of vortices exists in the Hamiltonian limit is
of particular interest in its own right. These themes will await further consideration.

\end{document}